%% file: main.tex
\documentclass[sigconf, nonacm]{acmart}

\usepackage{amsmath,amsfonts}
\usepackage{natbib}
\usepackage{array}
\usepackage[caption=false,font=normalsize,labelfont=sf,textfont=sf]{subfig}
\usepackage{textcomp}
\usepackage{stfloats}
\usepackage{url}
\usepackage{verbatim}
\usepackage{graphicx}
\usepackage{multirow}
\usepackage{svg}
\usepackage[linesnumbered,ruled,vlined]{algorithm2e}
\usepackage[table]{xcolor}
\usepackage{tikz}
\usetikzlibrary{trees, arrows.meta, positioning}
\definecolor{mygreen}{HTML}{2f9e44}
\definecolor{myred}{HTML}{e03131}
\definecolor{myblue}{HTML}{1971c2}
\definecolor{myorange}{HTML}{f08c00}

\newcommand{\activecell}[1]{\cellcolor{black!10}#1}

\newtheorem{definition}{Definition}
\newtheorem{example}{Example}

\newif\iffullversion
\fullversiontrue

\newcommand\vldbdoi{10.14778/xxxxxx.xxxxxx}
\newcommand\vldbpages{1111 - 1122}
\newcommand\vldbvolume{xx}
\newcommand\vldbissue{x}
\newcommand\vldbyear{2027}
\newcommand\vldbauthors{\authors}
\newcommand\vldbtitle{\shorttitle} 
\newcommand\vldbavailabilityurl{https://github.com/LukMRVC/tree_similarity_extension}
\newcommand\vldbpagestyle{empty}

\begin{document}

\title{Practical Threshold-based Tree Edit Distance Lower-Bounds}

\author{Lukáš Moravec, Radim Bača}

\affiliation{%
  \institution{VSB - Technical Universitty of Ostrava}
  \streetaddress{17.listopadu 15}
  \city{Ostrava}
  \state{Czech Republic}
  \postcode{70800}
}
\email{lukas.moravec.st2@vsb.cz,radim.baca@vsb.cz}



\begin{abstract}
Threshold-based similarity search over tree-structured data using tree edit distance (TED) is computationally intensive. Given a query tree and a database of trees, the goal is to retrieve all trees within a predefined TED threshold $\tau$. Because exact TED computation is expensive, practical methods employ lower-bounds to prune dissimilar candidates before verification.

Existing lower-bounds exhibit a fundamental trade-off: inexpensive statistical and structural bounds provide limited pruning power, whereas the more precise traversal-based string edit distance (SED) bound is expensive to compute using standard quadratic dynamic programming. Moreover, previous comparative studies do not cover recent structural filters or threshold-aware SED implementations, leaving their practical trade-offs unclear.

In this article, we first provide a comprehensive experimental comparison of state-of-the-art TED lower-bounds in terms of pruning precision and computational cost. We then accelerate the SED lower-bound using Ukkonen's bounded string edit distance algorithm, substantially reducing its runtime without affecting its pruning power. Finally, we introduce the SED-struct threshold filter, which strengthens SED with axes-aware constraints capturing structural relationships among tree nodes. Experiments on synthetic and real-world datasets show that SED-struct consistently achieves the highest filtering precision while retaining practical filtering costs. The results suggest that SED-struct is particularly beneficial for heterogeneous tree collections in which the standard SED lower-bound achieves relatively low precision.
\end{abstract}

\maketitle

\pagestyle{\vldbpagestyle}
\begingroup\small\noindent\raggedright\textbf{PVLDB Reference Format:}\\
\vldbauthors. \vldbtitle. PVLDB, \vldbvolume(\vldbissue): \vldbpages, \vldbyear.\\
\href{https://doi.org/\vldbdoi}{doi:\vldbdoi}
\endgroup
\begingroup
\renewcommand\thefootnote{}\footnote{\noindent
This work is licensed under the Creative Commons BY-NC-ND 4.0 International License. Visit \url{https://creativecommons.org/licenses/by-nc-nd/4.0/} to view a copy of this license. For any use beyond those covered by this license, obtain permission by emailing \href{mailto:info@vldb.org}{info@vldb.org}. Copyright is held by the owner/author(s). Publication rights licensed to the VLDB Endowment. \\
\raggedright Proceedings of the VLDB Endowment, Vol. \vldbvolume, No. \vldbissue\ %
ISSN 2150-8097. \\
\href{https://doi.org/\vldbdoi}{doi:\vldbdoi} \\
}\addtocounter{footnote}{-1}\endgroup

\ifdefempty{\vldbavailabilityurl}{}{
	\vspace{.3cm}
	\begingroup\small\noindent\raggedright\textbf{PVLDB Artifact Availability:}\\
	The source code, data, and/or other artifacts have been made available at \url{https://github.com/LukMRVC/tree_similarity_extension} \& \url{https://github.com/LukMRVC/ted-search}.
	\endgroup
}

\input{introduction}
\input{related_work}

\input{problem}
\input{lower_bounds}

\input{sed_struct}
\input{experiments}

\input{conclusion}
\input{acknowledgments}

\bibliographystyle{ACM-Reference-Format}
\bibliography{references}

\end{document}

%% file: introduction.tex
\section{Introduction}

The search for similar trees is an underlying problem that can be applied in linguistics~\cite{augustinus2012example} for syntactical analysis of treebank patterns, in bioinformatics~\cite{ma2002computing}~\cite{liu2006method} for comparing secondary RNA structures, in web data~\cite{kim2007web} to detect near-duplicate websites, or in code plagiarism detection~\cite{son2013application,song2024revisiting} that is based on an AST similarity.

The most common similarity measure for ordered labeled trees is the tree edit distance~\cite{tai1979tree} (TED), counting the minimum number of tree node edits such as node insertion, node relabel, and node deletion. 

Several important problems involve computing the tree edit distance under a threshold constraint, a setting we refer to as the \emph{bounded tree edit distance} problem. Typical examples include threshold-based similarity search, where the goal is to find all trees within a given TED from a query tree~\cite{pawlik2015efficient}, and similarity join, which aims to identify all tree pairs whose TED is below a specified threshold~\cite{hutter2019effective}. Due to the high computational cost of exact TED computation, the most efficient approaches adopt a filter-verify paradigm~\cite{pawlik2015efficient,hutter2019effective,guha2002approximate,li2014survey}. In this approach, a fast lower-bound is first applied to eliminate pairs of trees that cannot possibly satisfy the threshold, followed by an exact verification phase that computes the TED only for the remaining candidates.


A wide spectrum of lower-bound techniques for TED has been proposed~\cite{hutter2019effective,kailing2004efficient,yang2005similarity,guha2002approximate}.
The simplest bounds exploit only tree size or label multiset statistics~\cite{kailing2004efficient}, while more advanced approaches incorporate limited structural information, for instance through binary branch representations~\cite{yang2005similarity}, traversal-based encodings~\cite{guha2002approximate}, or tree decompositions~\cite{hutter2019effective}.
Among these techniques, the \emph{string edit distance} (SED) lower-bound~\cite{guha2002approximate}, obtained by linearizing trees via traversals and computing edit distance over the resulting strings, has repeatedly been shown to provide the highest pruning power~\cite{li2014survey}. However, this superior precision comes at a significant computational cost. The standard dynamic programming algorithm for SED requires quadratic time and space in the traversal length, which directly translates into quadratic complexity in the tree size. As a result, SED becomes a performance bottleneck for large tree collections or higher distance thresholds.

A further limitation of existing work lies in the lack of a comprehensive and reproducible experimental comparison of state-of-the-art TED lower-bounds. Although a survey of lower-bound techniques exists~\cite{li2014survey}, it primarily focuses on conceptual comparisons and asymptotic properties. Importantly, it does not consider efficient bounded implementations of SED, such as those based on Ukkonen-style algorithms, and lacks later advances in structural filtering. As a consequence, the conclusions drawn in that survey no longer reflect the current state-of-the-art. In particular, the more recent \emph{structural filter} lower-bound~\cite{hutter2019effective} is not covered. Structural filter exploits detailed positional and hierarchical information of tree nodes and has been shown to outperform all classical lower-bounds except SED.

This paper presents four key contributions to advance TED lower-bound filtering:
\begin{itemize}
    \item We perform a comprehensive experimental comparison of state-of-the-art lower-bounds, evaluating their trade-offs in precision and cost, including efficient bounded implementations not covered in prior surveys.
    \item We improve the efficiency of the SED lower-bound by incorporating Ukkonen’s bounded SED implementation. We demonstrate that SED-based filtering is superior to all existing lower-bound approaches when using the Ukkonen optimization, despite SED’s historically high computational cost.
    \item Most importantly, we introduce the SED-struct threshold filter, which combines SED with \emph{axes-aware} structural constraints. We prove the soundness of the filter and introduce an Ukkonen-based algorithm for its efficient bounded computation. This tailored approach achieves significantly higher precision than other filtering methods, without significant computational overhead.
    \item We make our implementation freely available and provide a PostgreSQL extension that makes the evaluated filtering methods available directly within a database system.
\end{itemize}

The axes-aware structural difference is the key technical novelty enabling the SED-struct filter. Extensive evaluation on both synthetic and real-world datasets confirms that our SED-struct filter achieves the best balance of precision and speed, making it a practical and effective component for threshold-based TED similarity search applications.

%% file: related_work.tex
\section{Related Work}

Research on tree edit distance (TED) has evolved along three complementary directions: exact TED algorithms, complexity-theoretic limits,
and threshold-oriented filtering for similarity search and joins.

\subsection{Exact TED Algorithms}

\iffullversion
Tai~\cite{tai1979tree} introduced the first polynomial-time algorithm for ordered TED.
Subsequent work significantly reduced practical and theoretical cost through dynamic programming refinements,
including the Zhang--Shasha decomposition~\cite{zhang1989simple}, the heavy-path decomposition of Klein~\cite{klein1998computing}, and robust strategy optimization
in RTED (and later APTED)~\cite{pawlik2015efficient}.
In 2021 Mao~\cite{mao::focs-2021} broke the cubic barrier
for unweighted version of TED using max-plus matrix product, resulting in $O(n^{2.9546})$. Nogler et al.~\cite{nogler:ted-apsp}
improved this to $O(n^{2.6857})$ while also proving fine-grained equivalence
to APSP.
To bridge the gap between these theoretical bounds and practical large-scale applications, Fan et al.~\cite{fan2024x} introduced X-TED, a GPU-accelerated execution framework that employs dependency-driven preprocessing and adaptive scheduling for dynamic programming tasks, significantly enhancing throughput and memory efficiency for extensive TED workloads.
\fi

\subsection{Bounded Tree Edit Distance}
\iffullversion
In many scenarios, the exact distance value is secondary to a simple binary check: whether TED between $T_1$ and $T_2$ is lower than a given threshold $\tau$. This requirement motivated bounded TED algorithms parameterized by $\tau$, which aim for a running time of $O(n + \text{poly}(\tau))$. These methods are significantly more efficient than standard approaches when the threshold $\tau$ is much smaller than the tree size $n$.

Touzet~\cite{touzet2005linear} introduced the first such algorithm with a complexity of $O(n\tau^3)$. Akmal et al.~\cite{akmal2021faster} later reduced this to $O(n\tau^2 \log n)$ for unweighted tree edit distance. Das et al.~\cite{das2023weighted} reached a major milestone by achieving a running time of $\tilde{O}(n + \text{poly}(\tau))$. Most recently, Kociumaka and Shahali~\cite{kociumaka2025faster} further improved both weighted and unweighted settings, proving a runtime of $O(n + k^6\log k)$.

While these theoretical results show that small $\tau$ values allow for much faster computation than the general $O(n^3)$ baseline, many of these advancements have yet to be fully realized in practice. For real-world applications, the TopDiff algorithm by Pawlik and Augsten~\cite{topdiff} remains the current state-of-the-art.
\fi

\subsection{TED Lower-bounds for Search and Join}

Bounded retrieval typically follows a filter-and-verify paradigm: compute cheap lower-bounds first,
then run exact TED only on surviving candidates~\cite{li2014survey,hutter2019effective}.
Early lower-bounds use coarse statistics such as size and label overlap~\cite{kailing2004efficient}.
Tighter lower-bounds exploit structural signals, including binary-branch representations~\cite{yang2005similarity} and traversal-based
string encodings~\cite{guha2002approximate}.
While the $pq$-gram distance~\cite{dataset:bolzano} is another fast way to filter trees, it only serves as a lower-bound for fanout-weighted TED and is therefore inapplicable to the widely used TED studied in this paper.
Prior comparative studies consistently report that the SED lower-bound achieves the highest precision,
but at the highest computational cost~\cite{li2014survey}.
Recent work has also improved the scalability of candidate generation itself.
Karpov and Zhang~\cite{karpov2023syncsignature} introduce SyncSignatures, a signature framework that enables efficient parallel candidate generation via hash joins,
substantially outperforming legacy single-threaded tree similarity join pipelines.
In addition, dedicated index structures for subtree similarity search have broadened the practical scope of TED-based retrieval,
enabling efficient discovery of matching hierarchical fragments in large forests.
This cost-quality trade-off motivates bounded SED implementations (e.g., Ukkonen-style pruning~\cite{ukkonen1985algorithms}) and drives our focus on
improving both efficiency and tightness.

Relative to prior surveys~\cite{li2014survey} and recent structural filters~\cite{hutter2019effective}, we provide an updated empirical comparison and contribute an enhanced SED-based
filter that combines bounded SED computation with additional structural information.

All empirical comparisons and evaluations in this paper use the unweighted (unit-cost) variant of TED.

%% file: problem.tex
\section{Problem Definition}



A tree $T$ is a rooted, directed acyclic graph with a node set $N(T)$ and an edge set $E(T) \subseteq N(T) \times N(T)$.
The tree has a $\mathrm{root}(T)$. Since $T$ is ordered, each node $u \in N(T)$ induces four disjoint subsets of $N(T)$.
These subsets are shown  in Figure~\ref{fig:tree_axes}, where we consider the node $b_{1}$.
Specifically, we are interested in the sizes of these subsets:

\begin{itemize}
  \item $\mathit{a}(b_{1})$ – the number of \textcolor{myred}{ancestors} of $b_{1}$,
  \item $\mathit{f}(b_{1})$ – the number of \textcolor{myblue}{following} nodes of $b_{1}$,
  \item $\mathit{d}(b_{1})$ – the number of \textcolor{myorange}{descendants} of $b_{1}$,
  \item $\mathit{p}(b_{1})$ – the number of \textcolor{mygreen}{preceding} nodes of $b_{1}$.
\end{itemize}

\begin{figure}[htb]
  \centering
    \includegraphics[width=0.25\textwidth]{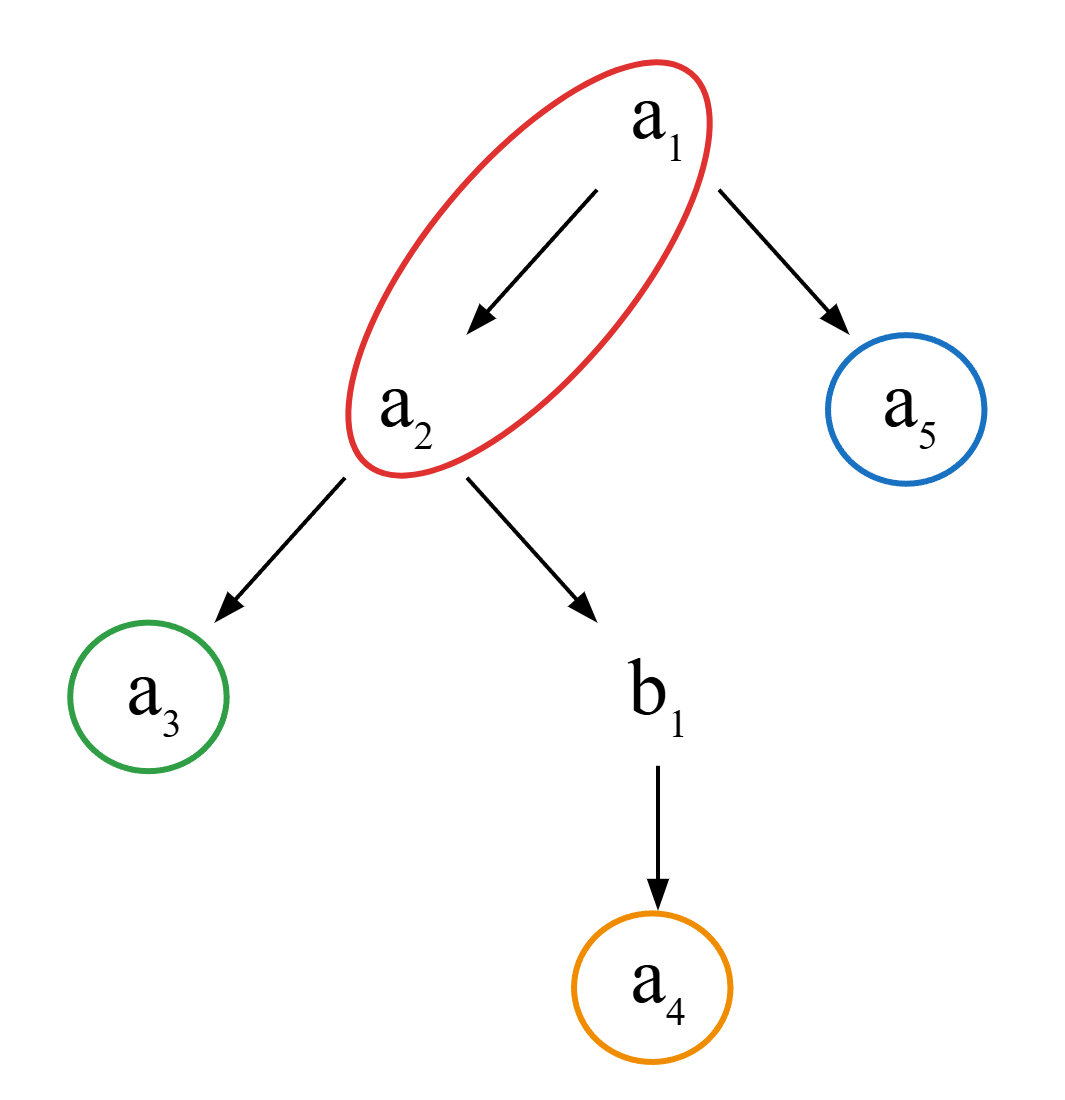}
  \caption{Disjoint subsets for $b_1$}
  \label{fig:tree_axes}
\end{figure}

Every node $u$ has a structural vector 
\[
\mathit{sv_{afdp}}(u) = (a(u), f(u), d(u), p(u))
\]

We sometimes call $\mathit{sv_{afdp}}(u)$ a \emph{complete structural vector} since it considers all four relationship axes.

In addition, for any subset of axes, we may define a \emph{partial structural vector}. For example,
\[
  \mathit{sv_{af}}(u) = (a(u), f(u)), \qquad \mathit{sv_{dp}}(u) = (d(u), p(u)).
\]
Partial structural vectors allow the capture of only selected relationship axes of a node. These structural vectors play a central role in the structural filter (Section~\ref{sec:lb-filters}) and in our SED-struct threshold filter (Section~\ref{sec:sedstructfilter}).

Every node $u$ has a label $lbl(u)$. $|T|$ denotes the number of nodes in a tree $T$.

Let $\delta(T_1, T_2)$ denote the tree edit distance between trees $T_1$ and $T_2$. It is defined as the minimum number of operations required to transform $T_1$ into $T_2$. The allowed operations are:
\begin{itemize}
  \item $\mathit{insert}(u, v, S)$ – insert a node $u$ as a child of node $v$, making a consecutive subsequence $S$ of $v$'s children the children of $u$ (preserving order),
  \item $\mathit{delete}(u)$ – delete node $u$ and reattach its children to the parent of $u$ (preserving order),
  \item $\mathit{relabel}(u, \ell)$ – change the label of node $u$ to $\ell$.
\end{itemize}

\begin{figure*}[htb]
  \centering
  \begin{tabular}{cp{0.2cm}cp{0.2cm}c}
    \begin{tikzpicture}[remember picture, baseline]
       
        \node[inner sep=0pt] (T1) {\includegraphics[width=0.24\textwidth]{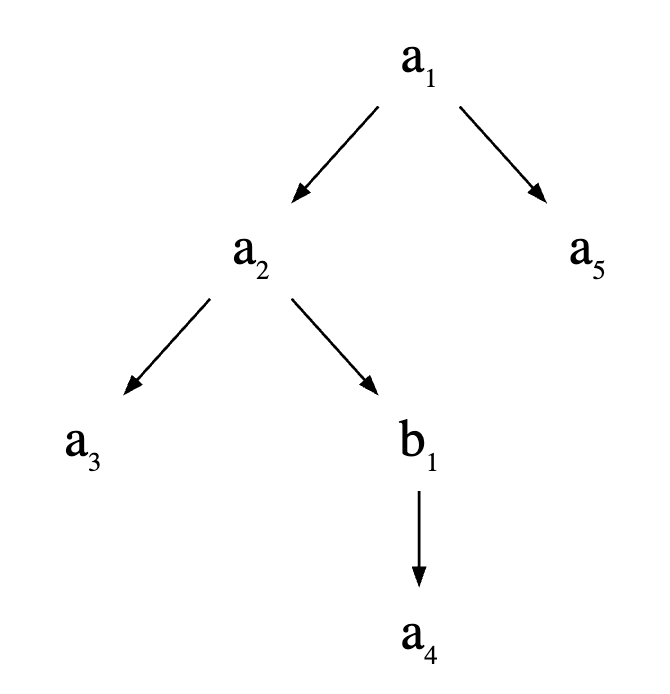}};
        \coordinate (T1-a1) at ([xshift=0.8cm, yshift=1.9cm]T1.center);
        \coordinate (T1-a2) at ([xshift=-0.3cm, yshift=0.65cm]T1.center);
        \coordinate (T1-a3) at ([xshift=-1.42cm, yshift=-0.52cm]T1.center);
        \coordinate (T1-a4) at ([xshift=0.75cm, yshift=-1.8cm]T1.center);
    \end{tikzpicture} & &
    
    \begin{tikzpicture}[remember picture, baseline]
        \node[inner sep=0pt] (T2) {\includegraphics[width=0.31\textwidth]{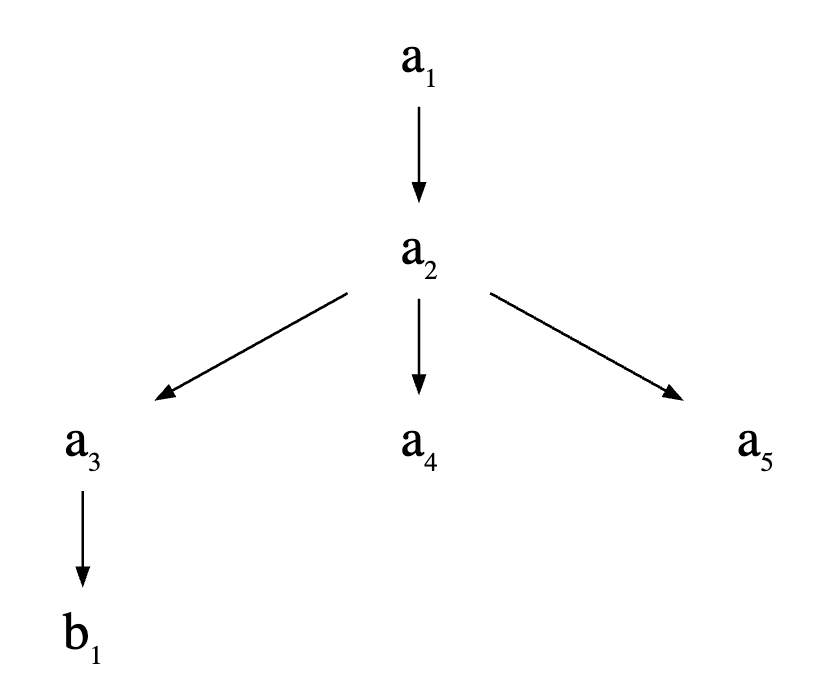}};
        
        \coordinate (T2-a1) at ([xshift=-0.3cm, yshift=2cm]T2.center);
        \coordinate (T2-b1) at ([xshift=-1.2cm, yshift=-1.2cm]T2.center);

        \coordinate (T2-a2) at ([xshift=-0.30cm, yshift=0.65cm]T2.center);
        \coordinate (T2-a3) at ([xshift=-2.45cm, yshift=-0.45cm]T2.center);
        \coordinate (T2-a4) at ([xshift=-0.30cm, yshift=-0.65cm]T2.center);
        
    \end{tikzpicture} & &
    \begin{tikzpicture}[remember picture, baseline]
        \node[inner sep=0pt] (T3) {\includegraphics[width=0.24\textwidth]{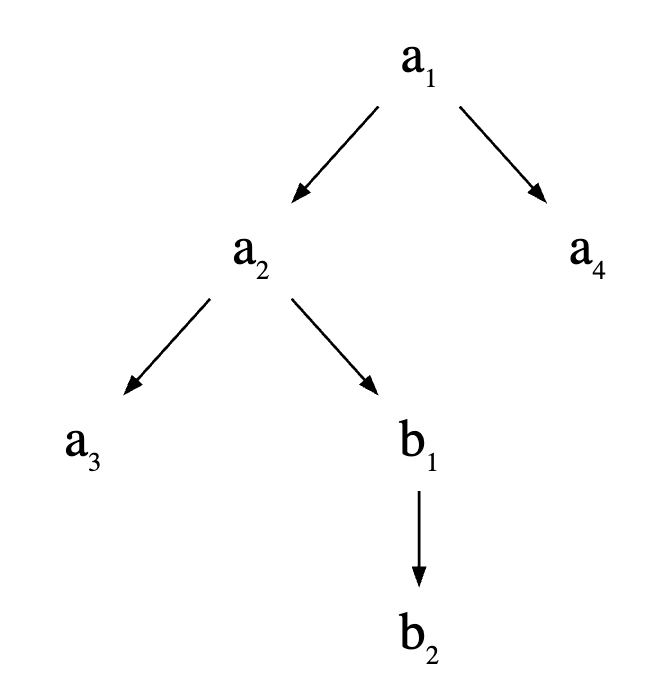}};
    \end{tikzpicture} \\
    
    T1 & & T2 & & T3 \\                                  
  \end{tabular}
  \begin{tikzpicture}[remember picture, overlay, >={Triangle[length=3.5pt, width=3.5pt]}]
    \draw[gray!60, dashed, -, thick] (T1-a2) to[bend left=20] (T2-a2);
    \draw[gray!60, dashed, -, thick] (T1-a1) to[bend left=20] (T2-a1);
    \draw[gray!60, dashed, -, thick] (T1-a3) to[bend left=20] (T2-a3);
    \draw[gray!60, dashed, -, thick] (T1-a4) to[bend left=3] (T2-a4);
  \end{tikzpicture}
  \caption{Sample trees}
  \label{fig:trees}
\end{figure*}
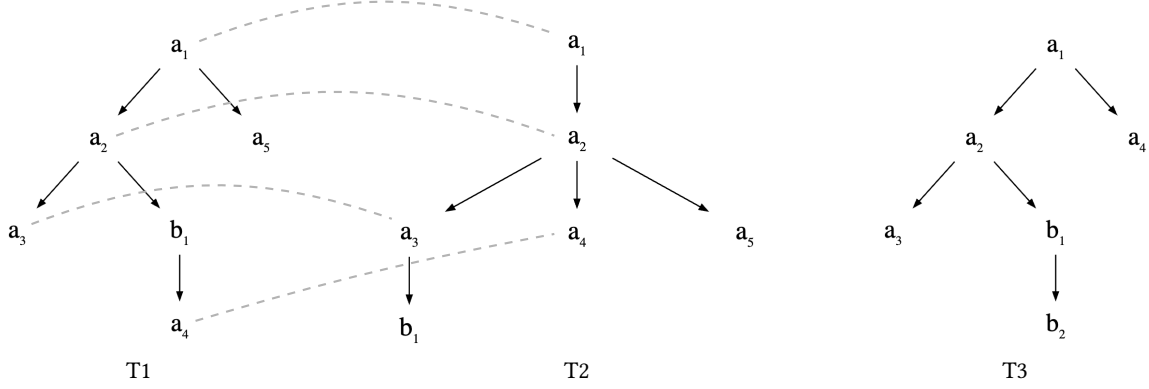

\begin{example}[Tree edit distance]
  \label{ex:problem}

  Consider trees in Figure~\ref{fig:trees}.
  The dashed lines between T1 and T2 show one optimal edit mapping: the four node pairs $a_1{\leftrightarrow}a_1$, $a_2{\leftrightarrow}a_2$, $a_3{\leftrightarrow}a_3$, and $a_4{\leftrightarrow}a_4$ are matched with identical labels (cost~$0$ each). The remaining nodes $b_1$ and $a_5$ in~T1 are deleted, while $b_1$ and $a_5$ in~T2 are inserted, giving a total cost of~$4$.
  The tree edit distance between T1 and T2 is therefore $\delta(T1,T2) = 4$, whereas $\delta(T1,T3) = 1$.
\end{example}

%% file: lower_bounds.tex
\section{Lower-Bound for Tree Edit Distance}
\label{sec:lb-filters}

In computing the edit distance between trees, many tree pairs can be quickly ruled out as dissimilar using lower-bounds before performing any expensive exact TED computation. A lower-bound computes a cheap approximation that is guaranteed to be \emph{no greater than} the true tree edit distance (TED); if this lower-bound exceeds a given threshold, the pair can be safely pruned. We discuss several fundamental lower-bounds for TED, following established terminology: the Size lower-bound, the Intersection lower-bound, the Binary branch lower-bound, and the Traversal String lower-bound (also known as the String Edit Distance lower-bound or SED lower-bound).

\subsection{Size Lower-Bound}

The simplest lower-bound for TED is based on the sizes of the trees. Intuitively, if one tree has more nodes than another, at least the difference in node count must be accounted for by edit operations (insertions or deletions).

\begin{definition}[Size Lower-Bound]
  \label{def:size-lb}
  \[
    LB_{size}(T_1, T_2) = \lvert |T_1| - |T_2| \rvert
  \]
\end{definition}

It can be computed in $O(1)$ time given the size of the trees, which makes it extremely fast, and we use it as the first lower-bound before any other one. However, it is often not very selective, since it ignores both the labels and the structural arrangement of the trees.

\subsection{Label intersection Lower-Bound}

A more refined lower-bound incorporates information on node labels. The label intersection lower-bound~\cite{kailing2004efficient} is based on the idea that if a node label appears in one tree but not the other, that mismatch will force at least one edit operation (either inserting a missing node, deleting an extra node, or renaming a mismatched node). We can formalise this by looking at the multiset of node labels in each tree and counting the overlap.

\begin{definition}[Label intersection Lower-Bound]
  \label{def:intersection-lb}

  Let $L_1$ and $L_2$ be the multisets of node labels in trees $T_1$ and $T_2$, respectively.
  \[
    LB_{int}(T_1, T_2) = max\{|T_1|, |T_2|\} - |L_1 \cap L_2|
  \]

  where $L_1 \cap L_2$ is the multiset intersection (common label occurrences in both trees).
\end{definition}

This lower-bound clearly improves upon the size bound by considering node labels. Nevertheless, the intersection lower-bound can still be quite coarse: it completely ignores tree structure beyond label frequencies. For example, assuming the trees in Figure~\ref{fig:trees}, the $LB_{int}(T1, T2) = 0$ whereas $LB_{int}(T1, T3) = 1$ suggests that the TED between T1 and T2 is lower than the TED between T1 and T3. However, it is the other way around, as mentioned in Example~\ref{ex:problem}.

\subsection{Binary Branch Lower-Bound}

The binary branch lower-bound~\cite{yang2005similarity} incorporates local structure by transforming a tree $T$ into its binary left-child, right-sibling representation $B_T$. Each node $u$ induces a branch $\mathrm{BB}(u)=(u,u_l,u_r)$, where a missing left or right child is represented by $\epsilon$. Let $\mathrm{BBV}(T)=(b_1,\ldots,b_N)$ count occurrences of the $N$ distinct branches in the dataset. For trees with vectors $(b_1,\ldots,b_N)$ and $(b'_1,\ldots,b'_N)$,
\[
  \mathrm{BBDist}(T_1,T_2)=\sum_{i=1}^{N}|b_i-b'_i|
  \leq 5\,\delta(T_1,T_2).
\]
Thus, $\mathrm{BBDist}(T_1,T_2)/5$ is a TED lower-bound. It is computable in $O(n)$ time but is generally loose.

\subsection{String Edit Distance (SED) Lower-Bound}

Another way to incorporate structural information is to linearize the trees into sequences of labels via a traversal, and then compare these sequences using string edit distance. The SED lower-bound~\cite{guha2002approximate} uses this approach: it converts each tree into a sequence (string) of node labels by performing a tree traversal (more on possible tree traversals in Section~\ref{sec:tree_traversal}), and then computes the classic string edit distance~\cite{masek1980faster} between the two resulting sequences. The key observation is that any edit operation on the tree will induce at most one edit operation on a traversal string of that tree. Therefore, the SED between the traversal strings will never exceed the edit distance between the trees. In formula form:

\begin{definition}[SED Lower-Bound]\label{def:traversal-lb}
  Let $\operatorname{pre}(T)$ and $\operatorname{post}(T)$ denote the preorder and postorder traversal strings of tree $T$, respectively (i.e., sequences of node labels obtained by listing nodes in preorder or postorder). Let $\mathrm{SED}(S_1, S_2)$ be the standard string edit distance between sequences $S_1$ and $S_2$. The SED lower-bound for two trees $T_1$ and $T_2$ is defined as:

  \[
    \begin{aligned}
      LB_{SED}(T_1, T_2) = max\{
       & SED(pre(T_1), pre(T_2)),    \\
       & SED(post(T_1), post(T_2))\}
    \end{aligned}
  \]

\end{definition}

Taking the maximum over both preorder and postorder traversals is important. A single traversal can sometimes miss structural differences – for instance, two distinct trees might yield identical preorder strings (SED = 0) even though their TED is non-zero, as long as their nodes appear in a coincident order by chance. In such cases, the postorder strings will typically differ and provide a non-zero edit distance. Thus, using both traversals ensures a more robust lower-bound.

The standard dynamic programming computation of SED constructs an edit distance matrix $D \in \mathbb{N}^{(m+1)\times(n+1)}$, where $m=|S_1|$ and $n=|S_2|$.
Entry $D[i,j]$ stores the edit distance between the prefixes $S_1[1..i]$
and $S_2[1..j]$. The recursion is given by
\[
  D[i,j] =
  \begin{cases}
    0, & i=j=0,            \\
    i, & j=0,              \\
    j, & i=0,              \\

    \begin{aligned}
      \min  \{ & D[i-1,j]+1,                                   \\
               & D[i,j-1]+1,                                   \\
               & D[i-1,j-1] +\mathbf{1}_{S_1[i]\neq S_2[j]}\}, \\
    \end{aligned}
       & \text{otherwise}.
  \end{cases}
\]

The $\text{SED}(S_1, S_2) = D[m, n]$.

Although computing an SED with an edit distance matrix is $O(n \cdot m)$ in the length of the sequences, this is still far cheaper than an exact tree edit distance computation in many cases (exact TED algorithms are often exponential or high-polynomial in the worst case). The SED lower-bound thus strikes a balance between tightness of the bound (by preserving some structural information) and efficiency of computation, making it a powerful tool in tree similarity queries.


\begin{example}[SED lower-bound]
  \label{ex:sed-example}

  Consider the trees $T_1$ and $T_2$ from Figure~\ref{fig:trees}.

  The preorder traversal strings are:
  \[
    \begin{aligned}
      \operatorname{pre}(T_1) = \mathtt{a_1\,a_2\,a_3\,b_1\,a_4\,a_5} \\
      \operatorname{pre}(T_2) = \mathtt{a_1\,a_2\,a_3\,b_1\,a_4\,a_5}
    \end{aligned}
  \]
  Hence,
  \[
    \mathrm{SED}(\operatorname{pre}(T_1), \operatorname{pre}(T_2)) = 0.
  \]

  The postorder traversal strings are:
  \[
    \begin{aligned}
      \operatorname{post}(T_1) & = \mathtt{a_3\;a_4\;b_1\;a_2\;a_5\;a_1} \\
      \operatorname{post}(T_2) & = \mathtt{b_1\;a_3\;a_4\;a_5\;a_2\;a_1}
    \end{aligned}
  \]
    \iffullversion
    The corresponding dynamic-programming matrix is:
    \[
      \renewcommand{\arraystretch}{1.1}
      \begin{array}{c|ccccccc}
            & - & b_1 & a_3 & a_4 & a_5 & a_2 & a_1        \\ \hline
        -   & 0 & 1   & 2   & 3   & 4   & 5   & 6          \\
        a_3 & 1 & 1   & 1   & 2   & 3   & 4   & 5          \\
        a_4 & 2 & 2   & 1   & 1   & 2   & 3   & 4          \\
        b_1 & 3 & 2   & 2   & 2   & 2   & 3   & 4          \\
        a_2 & 4 & 3   & 2   & 2   & 2   & 2   & 3          \\
        a_5 & 5 & 4   & 3   & 2   & 2   & 2   & 2          \\
        a_1 & 6 & 5   & 4   & 3   & 2   & 2   & \mathbf{2}
      \end{array}
    \]
    Therefore,
  \[
      \mathrm{SED}(\operatorname{post}(T_1), \operatorname{post}(T_2))
      = D[6,6] = 2,
  \]
    \fi
    Since the preorder distance is zero, $LB_{\mathrm{SED}}(T_1,T_2)=2$ and, consequently, $\delta(T_1,T_2) \ge 2$.
\end{example}

The SED lower-bound forms the foundation for our enhanced approach. In Section~\ref{sec:sedstructfilter}, we discuss algorithmic optimizations (including tree traversal strategies and bounded SED computation via Ukkonen’s algorithm) and introduce the SED-struct threshold filter, which augments SED with structural constraints.

\subsection{Structural Filter}

Another type of TED lower-bound that combines information about structure and labels is the \emph{structural filter}~\cite{hutter2019effective}.
Its definition builds upon the notion of an \emph{edit mapping}.

\begin{definition}[Edit Mapping]
  Let $T_1, T_2$ be two ordered labeled trees.
  An \emph{edit mapping} $M \subseteq N(T_1) \times N(T_2)$ is a set of node pairs $(v,w)$ such that:
  \begin{itemize}
    \item \textbf{One-to-one:} $v=v'$ iff $w=w'$ for any $(v,w),(v',w') \in M$.
    \item \textbf{Ancestor preservation:} $v$ is ancestor of $v'$ iff $w$ is ancestor of $w'$.
    \item \textbf{Order preservation:} $v$ is to the left of $v'$ iff $w$ is to the left of $w'$.
  \end{itemize}
\end{definition}

\begin{definition}[Edit-Mapping Cost]
  Let $M_1=\{v\mid(v,w)\in M\}$ and
  $M_2=\{w\mid(v,w)\in M\}$ be the projections of an edit mapping $M$.
  Under unit edit costs, its cost is
  \[
    \operatorname{cost}(M)
    = |N(T_1)\setminus M_1| + |N(T_2)\setminus M_2|
      + \sum_{(v,w)\in M}\mathbf{1}_{lbl(v)\neq lbl(w)}.
  \]
  The first two terms count deletions and insertions, respectively, and
  the last term counts relabelings. The standard edit-mapping
  characterization of TED gives
  \[
    \delta(T_1,T_2)=\min_M \operatorname{cost}(M),
  \]
  where the minimum ranges over all edit mappings between $T_1$ and $T_2$.
\end{definition}

A central idea is that only specific node pairs can belong to a valid mapping under a given threshold.
\begin{definition}[Structural Difference]
  \label{def:structural-difference}
  Let $v \in N(T_1)$ and $w \in N(T_2)$.
  The \emph{structural difference} between $v$ and $w$ is defined as the $L_1$ distance between their structural vectors:
  \[
    \Delta_{struct}(v,w) =
    \left\| \, sv_{afdp}(v) - sv_{afdp}(w) \, \right\|_1,
  \]
\end{definition}

Note that the subscript $afdp$ denotes the set of relationship axes taken into account when computing the structural difference.
The structural filter is defined with all four axes (i.e. it uses complete structural vectors). Still, in Section~\ref{sec:sedstructfilter} we demonstrate that considering only a subset of axes can be advantageous when combined with the SED lower-bound.

\begin{definition}[$\tau$-valid node pair]
  Let $T_1,T_2$ be two trees and $\tau$ a distance threshold.
  A node pair $(v,w)\in N(T_1)\times N(T_2)$ is called \emph{$\tau$-valid} if
  \[
    \begin{aligned}
      \Delta_{struct}(v,w) \;\leq \tau.
    \end{aligned}
  \]
  If additionally $\operatorname{lbl}(v)=\operatorname{lbl}(w)$, then $(v,w)$
  is called a \emph{$\tau$-valid label match}.
\end{definition}

Only $\tau$-valid label match node pairs can participate in an edit mapping of cost at most $\tau$.
This motivates the following definition.

\begin{definition}[$\tau$-valid label intersection]
  For trees $T_1,T_2$ and threshold $\tau$, the $\tau$-valid label intersection is
  \[
    \lvert N(T_1)\cap_\tau N(T_2) \rvert \;=\;
    \max_{P}\{\lvert P\rvert\},
  \]
  where $P$ is a one-to-one mapping between nodes of $T_1$ and $T_2$ such that all
  pairs in $P$ are $\tau$-valid label matches.
\end{definition}

Hütter et al.~\cite{hutter2019effective} describe an efficient algorithm for $\tau$-valid label intersection search for a pair of trees;
the implementation relaxes the strict one-to-one mapping to favour runtime at the cost of slightly looser bounds. For details, see their paper.

\begin{definition}[Structural Filter]
  For any two trees $T_1, T_2$ and threshold $\tau$,
  \[
    LB_{struct}(T_1, T_2, \tau) =  \max\big(\lvert T_1\rvert, \lvert T_2\rvert\big) - \lvert N(T_1)\cap_\tau N(T_2)\rvert.
  \]
\end{definition}

\begin{figure*}[htb]
  \centering
  \begin{tabular}{cp{0.2cm}c}
    \includegraphics[width=0.28\textwidth]{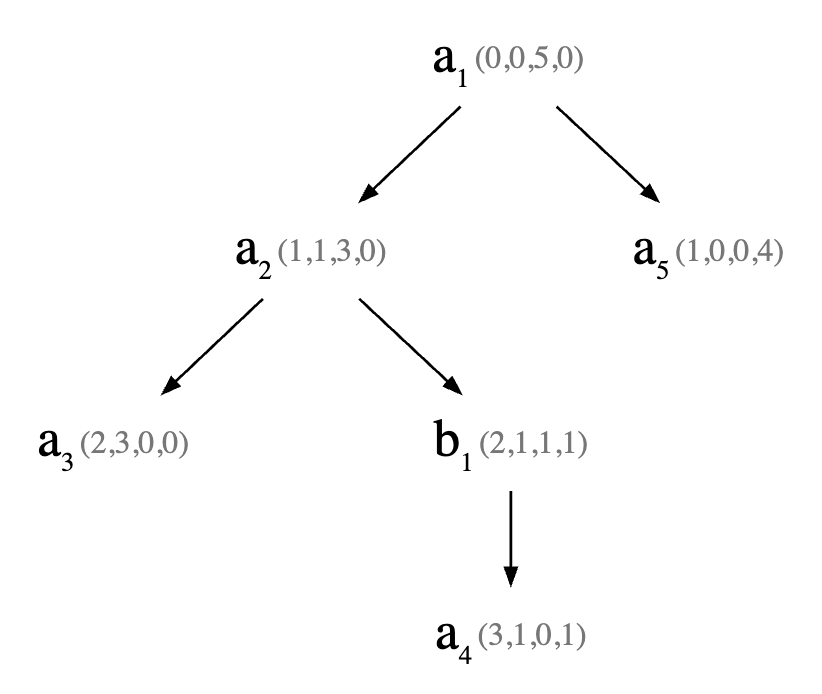} &  &
    \includegraphics[width=0.35\textwidth]{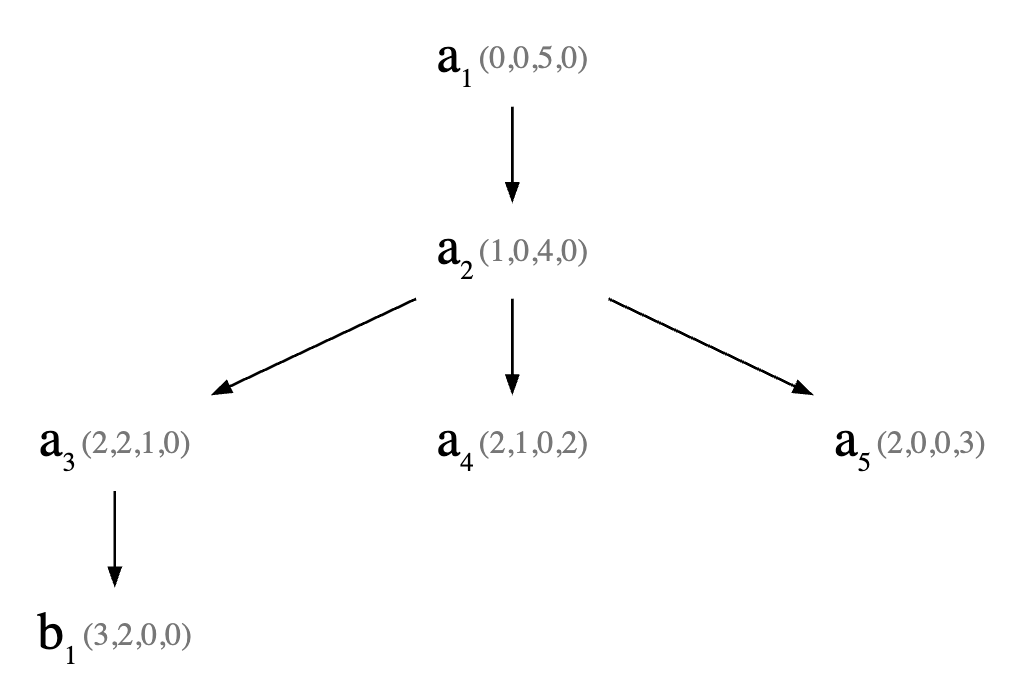}         \\
    T1                                                        &  & T2 \\
  \end{tabular}
  \caption{Sample trees from Figure~\ref{fig:trees} with afdp structural vectors}
  \label{fig:example_t1_t2_sl}
\end{figure*}

\begin{example}[Structural filter]

  Figure~\ref{fig:example_t1_t2_sl} contains trees T1 and T2 with structural vectors. We want to know whether their TED can be two or less; therefore, we compute $LB_{struct}(T1, T2, 2)$. The most important step is to find $\tau$-valid label intersection for $\tau = 2$. We compute structural difference between nodes with the same labels and find the following 2-valid label intersection: $(a_1, a_1)$, $(a_2,a_2)$, $(a_3, a_3)$, $(a_4, a_4)$, $(a_5, a_5)$. Clearly, the $LB_{struct}(T1, T2, 2) = 1$ and the structural filter pass the $\tau = 2$ threshold.
\end{example}

%% file: sed_struct.tex
\section{SED-struct Threshold Filter}
\label{sec:sedstructfilter}

The SED lower-bound is a highly robust method that achieves precision close to one across many real-world collections. However, in some collections, the structural variability of trees is higher, and the precision decreases. In this section, we introduce the SED-struct threshold filter, which augments SED with axes-aware structural constraints to achieve better precision in such cases. The computational overhead of the SED-struct filter is minimal relative to SED, making it the first usable filter of this kind that outperforms SED.

Let us begin by reviewing the key algorithmic enhancements that form the foundation of the SED-struct filter.

\subsection{Tree Traversal}
\label{sec:tree_traversal}

There are several different ways to linearize the tree into a sequence of labels. In this article we recognise a total of four traversals:

\begin{itemize}
  \item Preorder (pre)
  \item Postorder (post)
  \item Reverse preorder (rpre)
  \item Reverse postorder (rpost)
\end{itemize}
Here, reverse preorder and reverse postorder visit siblings from right to
left, while retaining the root-before-descendants and
root-after-descendants conventions, respectively.

For each traversal, two relationship axes occur before a node and the
complementary two occur after it, as summarized in
Table~\ref{tab:traversal-axes}.
\begin{table}[htb]
  \centering
  \caption{Relationship axes before and after a node for each traversal}
  \label{tab:traversal-axes}
  \begin{tabular}{c|c|c}
    $TT$ & Axes before the node & Axes after the node \\ \hline
    $\operatorname{pre}$   & $ap$ & $df$ \\
    $\operatorname{post}$  & $dp$ & $af$ \\
    $\operatorname{rpre}$  & $af$ & $dp$ \\
    $\operatorname{rpost}$ & $df$ & $ap$
  \end{tabular}
\end{table}

  \begin{definition}[Traversal-Specific Axis Set]
    \label{def:traversal-specific-axis-set}
    For a traversal order $TT$, let $X_{TT}$ be the pair of relationship
    axes occurring after a node in $TT$. Formally,
    \[
      X_{TT} =
      \begin{cases}
        df, & \text{if } TT \text{ is preorder},           \\
        af, & \text{if } TT \text{ is postorder},          \\
        dp, & \text{if } TT \text{ is reverse preorder},   \\
        ap, & \text{if } TT \text{ is reverse postorder}.
      \end{cases}
    \]
  \end{definition}

  Since an edit mapping preserves ancestor relationships and left-to-right
  order, mapped nodes occur in the same relative order in both traversal
  sequences for each of the four traversals.

In Section~\ref{sec:treetraversal}, we show that the combination of pre and post does not always have to be the best option.

\subsection{Ukkonen SED Algorithm}
\label{sec:ukkonen_sed}

\iffullversion
One possibility of optimizing expensive edit distance matrix computation is applying the Ukkonen SED algorithm~\cite{ukkonen1985algorithms, BerghelRoachASMExtension}, which takes
advantage of a given distance threshold $\tau$ and focuses only on the \emph{diagonal strip}.
This transformed approach utilizes diagonal indexes $k$, which are mapped from the matrix coordinates: $k = j - i$.
The main diagonal (also referred to as 0th diagonal) corresponds to $k=0$, the diagonal above it corresponds to $k=1$, the one below it corresponds to $k=-1$, and so on.

We use the extension of Ukkonen's algorithm introduced by Berghel and Roach \cite{BerghelRoachASMExtension}.
Their adjustment allows the algorithm to more precisely track the minimal-cost path within the dynamic programming matrix.
If string $S_1$ has length $m$ and string $S_2$ has length $n$ (where $m \geq n$), their
SED is at least $m - n$. The diagonal terminating at the lower-right entry of the $D$ matrix is called the \emph{target diagonal}.
The key idea is that an optimal alignment between $S_1$ and $S_2$
cannot deviate from the target diagonal by more than $\tau$.
Hence, only the diagonals $k \in \langle-\tau + (m - n),\tau + (m - n)\rangle$ need to be maintained. 
This method is further refined by the \emph{minimizing path condition}.

Algorithm~\ref{alg:ukkonen} presents the pseudocode of the extended Ukkonen's algorithm.

\begin{definition}[Ukkonen Matrix]
  Let $C \in \mathbb{Z}^{(\tau+1)\times(2\tau+1)}$ be the matrix used in Ukkonen's algorithm. Its row index $p\in\{0,\ldots,\tau\}$ is an edit budget, and its column index $k\in\{-\tau,\ldots,\tau\}$ identifies the diagonal $k=i-j$ of the SED matrix.
  Entry $C[p,k]$ stores the largest $X$ such that the prefixes $S_1[1..X+k]$ and $S_2[1..X]$ can be transformed into each other using at most $p$ edits. Thus, $X$ is the consumed prefix length of $S_2$, while $X+k$ is the consumed prefix length of $S_1$.
  The target diagonal is $k=m-n$. The algorithm terminates successfully when $C[p,m-n]=n$, equivalently when $X=|S_2|$ and $X+k=|S_1|$, which confirms that $\mathrm{SED}(S_1,S_2)\leq\tau$.
\end{definition}

The \emph{minimizing path condition} states that for every row $p' \in \langle0; \tau\rangle$ we
only need to compute diagonals $k$, which satisfy the following inequality:
\[
  \tau \geq p' + |m - n - k|
\]
At low values of $p'$ we know that we only need to compute
diagonals $k \in \langle p', p' + 1)$. However, as $p'$ increases, the minimizing path condition progressively
restricts the search space, narrowing the range of diagonals that must be considered.
Hence we compute $k_s$ (starting diagonal) and $k_e$ (ending diagonal)
in Algorithm~\ref{alg:ukkonen} on lines 6 and 7, respectively.

\begin{algorithm}
  \caption{Ukkonen Bounded String Edit Distance}
  \label{alg:ukkonen}
  \KwIn{Strings $S_1$ and $S_2$, threshold $\tau$}
  \KwOut{Boolean value whether $\mathrm{SED}(S_1,S_2) \leq \tau$}
  $m=|S_1|$ \;
  $n=|S_2|$ \;
  $k = m - n$ ; // $ m \geq n$ \\
  \lFor{$j \gets -\tau$ \KwTo $\tau$} {$C[-1, j] \gets -1$}
  \For{$p \gets 0$ \KwTo $\tau$}{
  $k_s = max(-p, k - (\tau - p))$ \;
  $k_e = min(p, k + (\tau - p))$ \;
  \For{$k' \gets k_s$ \KwTo $k_e$}{
  $\begin{aligned}
    X \gets \max \{ & C[p-1,k'-1],      \\
                     & C[p-1,k']+1,      \\
                     & C[p-1,k'+1]+1 \} \\
  \end{aligned}$
  \;
  \While{$S_1[X+k'+1] = S_2[X+1]$}{
  $X \gets X+1$ \;
  }
  $C[p,k'] \gets X$ \;
  \If{$X = n$ and $k' = k$}{
    \Return{true}\;
  }
  }
  }
  \Return{false}\;
\end{algorithm}

Intuitively, line 9 in Algorithm~\ref{alg:ukkonen} accounts for a substitution, insertion, or deletion, and lines 10 and 11 extend the match greedily along the diagonal. The algorithm runs in $O(\tau \cdot \min(|S_1|,|S_2|))$ time and $O(\tau)$ space, providing a substantial improvement over the quadratic dynamic programming algorithm.

\begin{example}[Ukkonen SED Algorithm]

  \label{ex:ukkonen}
  Let us apply Ukkonen's bounded SED algorithm (Algorithm~\ref{alg:ukkonen}) to the trees from Example~\ref{ex:sed-example} linearized by the postorder:
  \[
    \begin{aligned}
      S_1 = \operatorname{post}(T_1) = \mathtt{a_3\,a_4\,b_1\,a_2\,a_5\,a_1} \\
      S_2 = \operatorname{post}(T_2) = \mathtt{b_1\,a_3\,a_4\,a_5\,a_2\,a_1}
    \end{aligned}
  \]
  We choose the threshold $\tau = 2$, since the edit distance between the two sequences was computed as $2$.

  The algorithm maintains a matrix $C[p,k]$ of maximal prefix matches along diagonals. Initialization and subsequent updates produce the Ukkonen matrix $C$ in Figure~\ref{fig:SEDukkonen_ex}

  \begin{figure}
    \centering
    \[
      \begin{array}{c|ccc}
        p\backslash k & -1 & 0  & 1  \\ \hline
        -1            & -1 & -1 & -1 \\
        0             & -  & 0  & -  \\
        1             & 1  & 2  & 2  \\
        \mathbf{2}    & -  & 6  & -  \\
      \end{array}
    \]
    \caption{Ukkonen matrix $C$ for our running example}
    \label{fig:SEDukkonen_ex}
  \end{figure}

  The target diagonal is determined by the length difference $k = m - n = 0$. Hence, the algorithm ultimately checks whether the alignment on this diagonal reaches the end of the sequence.

  Consider the value $C[2,0]$. According to Algorithm~\ref{alg:ukkonen}, line~9 initializes
  $X$ from the three predecessor diagonals:

  \[
    X = \max\{C[1,-1],\; C[1,0] + 1,\; C[1,1] + 1\}.
  \]

  Using the values from the previous row of the matrix we obtain $X = 3$. Then, lines~10–11 extend the match greedily along the same diagonal, incrementing $X$ while $S_1[X+1] = S_2[X+j+1]$. In this case the suffixes match completely, so $X$ increases until $X = 6$, yielding $C[2,0] = 6$.

  Thus, Ukkonen's algorithm result states that the SED can be lower or equal to 2, which is consistent with the full dynamic-programming computation in Example~\ref{ex:sed-example}. However this algorithm uses only $O(\tau)$ space and $O(\tau \cdot \min(|S_1|, |S_2|))$ time.
\end{example}
\fi

\subsection{Structural Difference}

The structural filter introduced by Hutter et al.~\cite{hutter2019effective} uses the concept of structural difference to determine which node pairs can feasibly be matched within a given edit distance threshold. Definition~\ref{def:structural-difference} formalizes this concept.

While Hutter's approach uses the complete structural difference across all axes, our approach generalizes this by restricting to a subset of axes that is optimal for a given traversal order. This yields an \emph{axes-aware} variant:

\begin{definition}[Structural Difference (Axes-Aware Variant)]
  \label{def:axes-aware-structural-difference}
  For a subset of axes $X \subseteq \{a, d, f, p\}$, the \emph{axes-aware structural difference} is:
  \[
    \Delta_{struct}^{X}(v,w) =
    \left\| \, sv_{X}(v) - sv_{X}(w) \, \right\|_1,
  \]
\end{definition}

For the selected traversal $TT$, the SED-struct filter instantiates the
axes-aware structural difference with the traversal-specific axis set
$X_{TT}$ from Definition~\ref{def:traversal-specific-axis-set}, that is,
it uses $\Delta_{struct}^{X_{TT}}$.

The axes-aware variant allows us to combine structural differences with specific traversal orders, providing tighter bounds. For instance, a preorder traversal processes nodes in a depth-first manner, encoding ancestor and preceding axes in the sequence. By restricting the structural difference to the axes \emph{not} encoded by the traversal (descendants and followers for preorder), we capture only the structural constraints that the traversal order does not already enforce, leading to improved lower-bound precision without significant computational overhead.

\subsection{Combining SED with Structural Constraints}

Now we introduce the SED-struct filter, which combines the SED lower-bound with axes-aware structural constraints to achieve better precision. We first define the numerical score used by the filter.

\begin{definition}[SED-struct Score]
  \label{def:sed_struct_score}
  Let $S_1$ and $S_2$ be linearized node sequences of
  trees $T_1$ and $T_2$ obtained by a tree traversal $TT$, and let $\tau$ be a
  distance threshold.
  The SED-struct score is defined by a dynamic-programming matrix $D_{\tau}$,
  where entry $D_{\tau}[i,j]$ stores the guarded edit score between prefixes
  $S_1[1..i]$ and $S_2[1..j]$.
  The recursion is given by
  \[
    D_{\tau}[i,j] =
    \begin{cases}
      0, & i=j=0,            \\
      i, & j=0,              \\
      j, & i=0,              \\[6pt]

      \begin{aligned}
        \min \{
         & D_{\tau}[i-1,j] + 1,        \\
         & D_{\tau}[i,j-1] + 1,        \\
         & D_{\tau}[i-1,j-1] + c_{i,j}
        \},
      \end{aligned}
         & \text{otherwise},
    \end{cases}
  \]
  where structural guard
  \[
    c_{i,j} =
    \begin{cases}
      \infty, &
      \begin{aligned}
        \text{if } \Delta_{struct}^{X_{TT}}(x_i,y_j)
        + D_{\tau}[i-1,j-1] > \tau,
      \end{aligned}
      \\[1ex]
      0,      & \text{if } lbl(x_i) = lbl(y_j), \\
      1,      & \text{otherwise},
    \end{cases}
  \]
  , $X_{TT}$ is the traversal-specific axis set and $x_i$, and $y_j$ is a node in $S_1$, and $S_2$ respectively.
  
  The SED-struct score is $\text{SED-struct}(S_1,S_2,\tau)=D_{\tau}[m,n]$, where $m=|S_1|$ and $n=|S_2|$.
\end{definition}

\begin{lemma}[Prefix--Complement Accounting]
  \label{lem:prefix-complement-accounting}
  Let $M$ be an edit mapping, let $(x_i,y_j)\in M$, and let $A_M$ be the
  alignment induced by $M$ under traversal $TT$. If $q$ is the cost of
  the prefix of $A_M$ immediately before the diagonal step matching
  $x_i$ to $y_j$, then
  \[
    q+\Delta_{struct}^{X_{TT}}(x_i,y_j)
    \leq \operatorname{cost}(M).
  \]
\end{lemma}

\begin{proof}
  By the traversal-axis correspondence in
  Table~\ref{tab:traversal-axes}, the alignment prefix consumes exactly
  the nodes before $x_i$ and $y_j$. Its cost $q$ counts the unmatched
  nodes there and the label mismatches of mapped pairs there.

  Having an axis $z\in X_{TT}$, and let $r_z$ be the number of pairs in $M$
  whose first and second components lie on axis $z$ of $x_i$ and $y_j$,
  respectively. Because $M$ contains $(x_i,y_j)$ and preserves ancestor
  and order relationships, every mapped node on axis $z$ of either node
  is paired with a node on the same axis of the other node. Consequently,
  \[
    |z(x_i)-z(y_j)|
    \leq (z(x_i)-r_z)+(z(y_j)-r_z),
  \]
  where the right-hand side counts nodes on these two axes that are
  unmatched by $M$. The axes in $X_{TT}$ are disjoint, so summing this
  inequality over them shows that
  $\Delta_{struct}^{X_{TT}}(x_i,y_j)$ is at most the number of unmatched
  nodes on the axes after $x_i$ and $y_j$. These nodes are disjoint from
  all edit costs counted by $q$. Both sets of costs are included in
  $\operatorname{cost}(M)$, proving the claim.
\end{proof}

\begin{theorem}[Soundness of the SED-struct Threshold Filter]
  \label{thm:sed-struct-soundness}
  Under the unit-cost edit model, let $TT$ be one of the four traversal
  orders, and let $S_r = TT(T_r)$ for $r \in \{1,2\}$. For every
  $\tau\in\mathbb{N}_0$,
  \[
    \delta(T_1,T_2) \leq \tau
    \quad\Longrightarrow\quad
    \begin{aligned}
      \text{SED}(S_1,S_2)
      &\leq \text{SED-struct}(S_1,S_2,\tau) \\
      &\leq \delta(T_1,T_2) \leq \tau.
    \end{aligned}
  \]
  In particular,
  \[
    \text{SED-struct}(S_1,S_2,\tau) > \tau
    \quad\Longrightarrow\quad
    \delta(T_1,T_2) > \tau.
  \]
\end{theorem}

\begin{proof}
  Put $d = \delta(T_1,T_2)$ and suppose that $d \leq \tau$.
  Let $D$ denote the ordinary SED dynamic-programming matrix and write
  $\ell_{i,j}=\mathbf{1}_{lbl(x_i)\neq lbl(y_j)}$ for its diagonal cost.
  First, $D_{\tau}[i,j] \geq D[i,j]$ for all $i,j$. The boundary entries
  agree. For an interior entry, assuming the inequality for its three
  predecessors, the SED-struct recurrence can only replace the ordinary
  diagonal cost $\ell_{i,j}$ by $\infty$; monotonicity of $\min$ therefore
  proves the inequality by induction on $i+j$.
  Consequently,
  $\text{SED}(S_1,S_2) \leq \text{SED-struct}(S_1,S_2,\tau)$.

  Let $M$ be an optimal edit mapping, so $\operatorname{cost}(M)=d$.
  Since edit mappings preserve traversal order, $M$ induces an alignment
  $A$ of $S_1$ and $S_2$: unmatched nodes yield horizontal or vertical
  steps, and each matched pair yields a diagonal step whose cost is $0$
  or $1$ according to its labels. Hence, $\operatorname{cost}(A)=d$.

  Next, consider prefixes of $A$. For a prefix ending at $(i,j)$ with
  cost $q$, let the invariant be that all of its steps are admitted by the
  SED-struct recurrence and that $D_{\tau}[i,j] \leq q$. The base case is
  the empty prefix: it ends at $(0,0)$, has cost $0$, and satisfies
  $D_{\tau}[0,0]=0$.

  For the induction step, let a prefix ending at $(i',j')$ have cost $q$
  and satisfy the invariant. If the next alignment step consumes only one
  node, then it is horizontal or vertical, which is always admitted, and
  the recurrence gives
  \[
    D_{\tau}[i,j] \leq D_{\tau}[i',j'] + 1 \leq q+1.
  \]
  This is the cost of the extended prefix.

  It remains to consider a diagonal step, for which $(i',j')=(i-1,j-1)$
  and $(x_i,y_j)\in M$. Its extended-prefix cost is $q+\ell_{i,j}$.
  Lemma~\ref{lem:prefix-complement-accounting} gives
  \[
    q + \Delta_{struct}^{X_{TT}}(x_i,y_j) \leq d.
  \]
  The induction hypothesis is
  $D_{\tau}[i-1,j-1] \leq q$. Combining it with the preceding inequality
  yields the guard condition
  \[
    D_{\tau}[i-1,j-1] + \Delta_{struct}^{X_{TT}}(x_i,y_j)
    \leq q + \Delta_{struct}^{X_{TT}}(x_i,y_j)
    \leq d \leq \tau.
  \]
  Thus the guard does not block this diagonal, so $c_{i,j}=\ell_{i,j}$.
  The recurrence then establishes the induction conclusion for the
  extended prefix:
  \[
    D_{\tau}[i,j]
    \leq D_{\tau}[i-1,j-1] + \ell_{i,j}
    \leq q + \ell_{i,j}.
  \]

  By induction, the invariant holds for all prefixes of $A$. For the full
  alignment, this gives
  \[
    \text{SED-struct}(S_1,S_2,\tau)
    = D_{\tau}[m,n]
    \leq \operatorname{cost}(A)
    \leq d.
  \]
  Combining the two inequalities proves the first statement; the second
  follows by contraposition.
\end{proof}

\paragraph{Guard Semantics.} The structural guard is a necessary
feasibility test, not a sufficient one. By
Theorem~\ref{thm:sed-struct-soundness}, it never blocks the diagonal
step of an edit mapping whose total cost is at most $\tau$. Conversely,
if the guard blocks a diagonal step, that step may be discarded when
searching for an edit mapping of cost at most $\tau$. Passing the guard
does not by itself establish that the two nodes can be matched within
$\tau$, since label, order, and other structural constraints may still
make the full mapping infeasible.

The guard becomes progressively tighter as the alignment cost grows:
early in the alignment, when $D_{\tau}[i-1,j-1]$ is small, a moderately
large structural difference may still be tolerated; later, when most of
the budget has been spent, even a small structural mismatch can block the
diagonal transition.

\begin{definition}[SED-struct Threshold Filter]
  For any two trees $T_1,T_2$, tree traversals $TT_1,TT_2$, and threshold $\tau$, define the aggregate SED-struct score as
  \[
    \begin{aligned}
      \operatorname{score}_{\text{SED-struct}}(T_1,T_2,\tau)
      = \max\{&\text{SED-struct}(TT_1(T_1),TT_1(T_2),\tau),\\
              &\text{SED-struct}(TT_2(T_1),TT_2(T_2),\tau)\}.
    \end{aligned}
  \]
  The SED-struct threshold filter accepts the pair if this score is at
  most $\tau$ and rejects it otherwise.
\end{definition}

By Theorem~\ref{thm:sed-struct-soundness}, the filter is sound:
\[
  \operatorname{score}_{\text{SED-struct}}(T_1,T_2,\tau) > \tau
  \quad\Longrightarrow\quad
  \delta(T_1,T_2) > \tau.
\]
When $\delta(T_1,T_2) \leq \tau$, it satisfies
$\operatorname{score}_{\text{SED-struct}}(T_1,T_2,\tau)
\leq \delta(T_1,T_2)$. Since the score depends on $\tau$, no unrestricted
numerical lower-bound claim is made when $\delta(T_1,T_2)>\tau$; only the
filter's threshold decision is used.

The only traversal combination proposed in prior work is pre/post~\cite{guha2002approximate}. In Section~\ref{sec:treetraversal}, we show that the most effective combination for the SED-struct filter is rpre/pre.

\begin{example}[SED-struct on postorder, $\tau=3$]\label{ex:sed-struct-post}

  Let the postorder traversals be

  \[
    \begin{aligned}
      \operatorname{post}(T1) & = \mathtt{a_3\;a_4\;b_1\;a_2\;a_5\;a_1} \\
      \operatorname{post}(T2) & = \mathtt{b_1\;a_3\;a_4\;a_5\;a_2\;a_1}
    \end{aligned}
  \]

  \begin{figure}[bth]
    \centering
    \[
      \begin{tabular}{cc|ccccccc}
              &     &        & (3,2)          & (2,2)          & (2,1)  & (2,0)  & (1,0)  & (0,0)  \\
              &     & -      & $b_1$            & $a_3$            & $a_4$    & $a_5$    & $a_2$    & $a_1$    \\ \hline
              & -   & 0      & 1              & 2              & 3      & $\times$ & $\times$ & $\times$ \\
        (0,3) & $a_3$ & 1      & \activecell{2} & \activecell{3} & $\times$ & $\times$ & $\times$ & $\times$ \\
        (3,1) & $a_4$ & 2      & 2              & \activecell{3} & $\times$ & $\times$ & $\times$ & $\times$ \\
        (2,1) & $b_1$ & 3      & \activecell{3} & 3              & $\times$ & $\times$ & $\times$ & $\times$ \\
        (1,1) & $a_2$ & $\times$ & $\times$         & $\times$         & $\times$ & $\times$ & $\times$ & $\times$ \\
        (1,0) & $a_5$ & $\times$ & $\times$         & $\times$         & $\times$ & $\times$ & $\times$ & $\times$ \\
        (0,0) & $a_1$ & $\times$ & $\times$         & $\times$         & $\times$ & $\times$ & $\times$ & $\times$ \\
      \end{tabular}
    \]
    \caption{$D_{\tau}$ matrix for postorder linearized trees from Figure~\ref{fig:example_t1_t2_sl}}
    \label{fig:dtau_matrix}
  \end{figure}

For each node $u$ we attach its partial structural vector $\mathrm{sv}_{af}(u)$.
The dynamic-programming matrix $D_{\tau}$ is shown in Figure~\ref{fig:dtau_matrix} \emph{truncated by the threshold} $\tau=3$. A diagonal transition is allowed when the structural guard is satisfied; its cost is $0$ when the labels are equal and $1$ otherwise.
Cells where $c$ is computed as $\infty$ (they do not pass the structural guard) are highlighted using grey.

For example, in cell $D_{\tau}[a_3,b_1]$ the structural difference evaluates to $\Delta_{struct}^{af}(a_3,b_1) = 4$. Since the
threshold in this example is $\tau = 3$, the structural difference value exceeds the
error budget allowed ($c > \tau$) and the guard is not satisfied. Consequently, the candidate alignment
represented by this diagonal cannot correspond to a feasible edit path
within the threshold. The dynamic programming recurrence therefore ignores
this transition and the value of $D_{\tau}[a_3,b_1]$ is determined only from the
remaining two admissible predecessor cells.

Therefore, computation ends very quickly with a result `above threshold'.
We display $D_{\tau}[i,j]$ only if $D_{\tau}[i,j]\le \tau$; otherwise we print a cross ``$\times$''.

  Thus, the SED-struct filter is the only evaluated filter capable of establishing that $\delta(T1,T2)>3$.

\end{example}

\subsection{Ukkonen SED-struct Algorithm}


We extend Ukkonen's SED algorithm with a structural guard based on the structural difference $\Delta_{struct}^{X_{TT}}(\cdot,\cdot)$ introduced in Section~\ref{sec:treetraversal} and in Section~\ref{sec:sedstructfilter}.

Our Ukkonen SED-struct algorithm uses the matrix $C$ in the same manner as \iffullversion Algorithm~\ref{alg:ukkonen}\fi.
Each $C$ cell carries, in addition to the longest matched prefix length, a boolean flag \texttt{stop} indicating whether the path used to obtain this cell's value violated the structural guard.
While greedily extending matches along a diagonal $k$, we stop it when
$\Delta_{struct}^{X_{TT}}(S_1[X{+}k{+}1], S_2[X{+}1]) + p > \tau$ (structural guard), resolves to \texttt{false}, and refuse to increase $X$ even if the labels are equal (see line 16 and 17).
Moreover, when computing $C[p,k]$, if the vertical predecessor $C[p{-}1,k]$ has \texttt{stop} flag equal to \texttt{true}, we do not increase the matched prefix length by $+1$ to that predecessor (see line 12).

\begin{algorithm}
  \caption{Ukkonen SED-struct Algorithm}
  \label{alg:adjusted_ukkonen}
  \KwIn{Strings $S_1$ and $S_2$, threshold $\tau$}
  \KwOut{Boolean value whether the SED-struct score is at most $\tau$}

  \BlankLine
  $m \gets |S_1|$ \;
  $n \gets |S_2|$ \;
  $k \gets m-n$ ; // $m \geq n$ \

  \For{$k \gets -\tau$ \KwTo $\tau$}{
    $C[-1,k] \gets -1$ \;
    $C[-1,k].stop \gets false$ \;
  }

  \For{$p \gets 0$ \KwTo $\tau$}{
  $k_s = max(-p, k - (\tau - p))$ \;
  $k_e = min(p, k + (\tau - p))$ \; 
  \For{$k' \gets k_s$ \KwTo $k_e$}{
  $C[p,k'].stop \gets false$ \;

  $\begin{aligned}
    X \gets \max \{ & C[p-1,k'-1],                                \\
                    & C[p-1,k']+\mathbf{1}_{\neg C[p-1,k'].stop}, \\
                    & C[p-1,k'+1]+1 \} \\
  \end{aligned}$
  \;

  \If{\textbf{not} \textsc{StrGuard}$(p,k',X,\tau,S_1,S_2,C)$}{
    $C[p, k'] \gets X$ \;
    continue \;
  }

  \While{$S_1[X+ k' +1] = S_2[X+1]$}{
  \If{\textbf{not} \textsc{StrGuard}$(p,k',X,\tau,S_1,S_2,C)$}{
    break \;
  }
  $X \gets X+1$ \;
  }

  $C[p,k'] \gets X$ \;

  \If{$X = n$ \textbf{and} $k' = k$}{
    \Return{true}\;
  }
  }
  }
  \Return{false}\;

  \BlankLine
  \textbf{Function} \textsc{StrGuard}$(p,k',X,\tau,S_1,S_2,C)$:\;
  \Indp

  \If{$\Delta_{struct}^{X_{TT}}\!\bigl(S_1[X+k'+1],S_2[X+1]\bigr) + p > \tau$}{
  $C[p,k'].stop \gets true$ \;
  \Return{false}\;
  }
  \Return{true}\;
\end{algorithm}

\begin{example}[Ukkonen SED-struct algorithm]

  Let us demonstrate Ukkonen bounded SED-struct algorithm (Algorithm~\ref{alg:adjusted_ukkonen}) on our running example (i.e. postorder linearized trees T1 and T2 from Figure~\ref{fig:example_t1_t2_sl}) with threshold $\tau = 3$.

  The algorithm functionality is best described by the matrix $C$ in Figure~\ref{fig:adjukkonenmatrix}.
  Cells whose diagonal computation is stopped by the structural guard (i.e., with \texttt{stop} $=\mathrm{true}$) are highlighted in the matrix with grey shading.

  For instance, consider the cell $C[0,0]$ in the matrix $C$. In this case, the call to $\texttt{StrGuard}$ evaluates its condition as true, indicating that the structural difference between the corresponding nodes already exceeds the threshold $\tau$. Consequently, the algorithm sets $C[0,0].stop \gets \textit{true}$ and the diagonal extension step is not performed, i.e., the variable $X$ remains equal to $0$ (lines 14-15). Intuitively, this means that no prefix of the two sequences can be matched along this diagonal under the current error budget. The same structural guard condition also prohibits the diagonal alignment in the cell $D_{\tau}[a_3,b_1]$ as we mentioned in Example~\ref{ex:sed-struct-post}.

  \begin{figure}
    \centering
    \[
      \begin{array}{c|ccc}
        p\backslash k & -1 & 0              & 1              \\ \hline
        -1            & -1 & -1             & -1             \\
        0             & -  & \activecell{0} & -              \\
        1             & 1  & \activecell{0} & \activecell{0} \\
        2             & 2  & \activecell{1} & \activecell{0} \\
        3             & -  & 2              & -              \\
      \end{array}
    \]
    \caption{Adjusted Ukkonen matrix}
    \label{fig:adjukkonenmatrix}
  \end{figure}

  We observe that many fields in the matrix have a stop value of true.
  The Ukkonen matrix $C$ in Figure~\ref{fig:adjukkonenmatrix} corresponds to the $D_{\tau}$ from Example~\ref{ex:sed-struct-post} where the row reached on the main diagonal is 2.
  Recall that the column 0 corresponds to the target diagonal and the TED is 4 in our running example. Our Ukkonen bounded SED needs to reach the row 6 in column 0 at least in $C[3,0]$. However, it does not. Therefore, the algorithm rejects string similarity within $\tau =  3$ bound.

\end{example}

%% file: experiments.tex
\section{Experimental Evaluation}

This section reports on the experimental evaluation of TED filtering methods, comparing current state-of-the-art lower-bounds from prior work with our novel SED-struct threshold filter. Our evaluation includes:
\begin{enumerate}
  \item algorithmic efficiency when using Ukkonen's optimized algorithm for SED-based lower-bounds,
  \item impact of different tree traversals for SED-based methods.
  \item lower-bound quality with holistic end-to-end search queries measuring lower-bound runtime, precision and overall query runtime,
\end{enumerate}

All experiments employ the \emph{filter-verify framework}, a standard approach in similarity search.
In the filtering phase, a lower-bound is used to generate a set of candidate trees that could potentially satisfy the query.
Subsequently, in the verification phase, the exact TED for each candidate is computed using the \texttt{TopDiff} algorithm~\cite{topdiff}.
All evaluated filtering methods implement a trivial size lower-bound as a preliminary check before the actual computation begins.

To ensure measurement stability, each experiment was executed three times, and the best (minimum) runtime is reported. This methodology is standard practice in systems research, where reporting the best of multiple trials reduces the influence of unexpected system waits, context switches, and other transient performance variations. 

The lower-bounds were implemented in \texttt{Rust}~\footnote{\url{https://github.com/LukMRVC/ted-search}}.
Verification algorithm TopDiff is available from~\cite{Pawlik_Tree_Similarity_2020}.
All experiments were evaluated on an AMD Ryzen 5800X processor with 32GB of RAM.




\subsection{Datasets}
\label{subsec:datasets}
We used a total of 7 real-world datasets from different domains, such as 
bioinformatics (Swissprot \cite{dataset:swissprot}, RNA \cite{dataset:rna}, Treefam \cite{dataset:treefam}), bibliography (DBLP \cite{dataset:dblp}), NLP (PTB \cite{dataset:ptb}, Sentiment \cite{dataset:sentiment}) and code structures (Python \cite{dataset:python}).

While real-world datasets provide authentic evaluation scenarios for performance benchmarks,
they do not allow us to isolate and measure
the impact of specific properties on lower-bound performance.
To systematically assess how individual factors affect lower-bound effectiveness of different algorithms,
we created three collections of synthetic datasets:
\textit{Synthetic\textsubscript{size}},
\textit{Synthetic\textsubscript{struct}}, and
\textit{Synthetic\textsubscript{label}},
where key properties can be controlled independently\cite{dataset:moravec-baca}.
This controlled experimental approach enables us to understand the behavior of lower-bounds under varying
conditions that would be difficult to achieve with real data alone.

\iffullversion
\subsubsection{Synthetic\textsubscript{size}}
\label{subsubsec:synthetic_size_collection}
This collection of datasets varies the primary factor of tree size while keeping structure and labels constant.
By systematically increasing tree sizes across a predefined range,
we can examine how computational performance scales, allowing us
to examine the impact of Ukkonen's optimized algorithm.

\subsubsection{Synthetic\textsubscript{struct}}
\label{subsubsec:synthetic_structure_collection}
To isolate the effect of tree structure, this collection varies the average fanout (out-degree)
while maintaining constant tree size and label count. A high fanout produces shallow and bushy tree structures, whereas a
low fanout generates narrow and deep hierarchies.
These structural differences affect how effectively structural information can improve lower-bound precision.
We systematically decrease the average fanout to alter the resulting tree structures in datasets.

\subsubsection{Synthetic\textsubscript{label}}
\label{subsubsec:synthetic_label_collection}
The final collection of datasets examines the impact of label diversity on lower-bound effectiveness
by varying the pool of available distinct labels (from 10 down to 2) while keeping tree structure and size constant.
When labels are sparse, structural similarity becomes a more critical signal for edit distance approximation,
allowing us to isolate the benefit of incorporating structural information into lower-bounds.
\fi

The summarized characteristics of all synthetic datasets are in Table~\ref{tab:datasets}.

\subsubsection{Synthetic Dataset Generation}
The three synthetic collections were generated using a copy-and-modify approach to ensure controlled variance across experiments. The process begins with the \textbf{base generation} phase, where an initial
\textit{seed} tree is constructed according to specific parameters: node count, label pool size, and target average fanout.

Following the base generation, we employ \textbf{iterative modification} to populate the collection.
Each new tree is derived from the seed by applying a sequence of random modifications.
These operations consist of selecting nodes at random to perform either a \textit{child swap} or a \textit{subtree move}.
During this phase, we monitor the properties of the tree to ensure that the intended parameters remain stable.

\begin{table*}[ht]
  \centering
  \renewcommand{\arraystretch}{1.25}
  \begin{tabular}{|l|lrrrrr|} 
    \hline
    \textbf{Collection} & \textbf{Dataset} & \textbf{\#Trees} & \textbf{Avg. size} & \textbf{Max. size} & \textbf{\#Labels} & \textbf{Mean Fanout} \\ 
    \hline
    \multirow{7}{2cm}{\textit{Real}} 
                        & Sentiment        & 9645             & 36                 & 102                & 19468             & 1.92                 \\
    \cline{2-7}
                        & DBLP             & 150000           & 26                 & 1703               & 992865            & 1.85                 \\
    \cline{2-7}
                        & Treefam          & 5000             & 2666               & 15065              & 1276006           & 1.75                 \\
    \cline{2-7}
                        & Swissprot        & 250000           & 428                & 20490              & 3809254           & 1.79                 \\
    \cline{2-7}
                        & PTB              & 3832             & 72                 & 711                & 12966             & 1.54                 \\
    \cline{2-7}
                        & RNA              & 16819            & 145                & 317                & 170               & 2.52                 \\

    \cline{2-7}
                        & Python           & 50000            & 948                & 43270              & 1479589           & 1.59                 \\
    \hline

    \multirow{4}{2cm}{\textit{Synthetic\textsubscript{size}}}
                        & S10              & 1000             & 7.8                & 14                 & 1180              & 2.36                 \\
    \cline{2-7}
                        & S50              & 1000             & 50.8               & 53                 & 3754              & 2.59                 \\
    \cline{2-7}
                        & S100             & 1000             & 100.0              & 106                & 3700              & 2.83                 \\
    \cline{2-7}
                        & S1000            & 1000             & 1000.0             & 1006               & 3978              & 2.29                 \\
    \hline

    \multirow{5}{2cm}{\textit{Synthetic\textsubscript{struct}}}
                        & HF               & 5000             & 59                 & 59                 & 10                & 4.46                 \\
    \cline{2-7}
                        & MF               & 5000             & 59                 & 59                 & 10                & 3.41                 \\
    \cline{2-7}
                        & SF               & 5000             & 61                 & 61                 & 10                & 2.64                 \\
    \cline{2-7}
                        & LF               & 5000             & 58                 & 58                 & 10                & 2.21                 \\
    \cline{2-7}
                        & VLF              & 5000             & 61                 & 61                 & 10                & 1.92                 \\
    \hline

    \multirow{5}{2cm}{\textit{Synthetic\textsubscript{label}}}
                        & L10              & 5000             & 58                 & 58                 & 10                & 2.07                 \\
    \cline{2-7}
                        & L6               & 5000             & 60                 & 60                 & 6                 & 2.07                 \\
    \cline{2-7}
                        & L4               & 5000             & 60                 & 60                 & 4                 & 2.00                 \\
    \cline{2-7}
                        & L3               & 5000             & 62                 & 62                 & 3                 & 2.10                 \\
    \cline{2-7}
                        & L2               & 5000             & 59                 & 59                 & 2                 & 1.99                 \\

    \hline
  \end{tabular}
  \caption{Tree collections used in the experiments.}
  \label{tab:datasets}
\end{table*}

\subsection{Query selection}
\label{subsubsec:query-selection}

For each real-world dataset, we selected 100 distinct queries whose result sets contain approximately $1\%$ of the collection. This result rate represents searches for relatively similar trees while providing enough positive matches for stable precision measurements. We deliberately excluded queries that the size lower-bound alone can prune effectively, because such queries reveal little about the relative quality of the nontrivial filters evaluated here. The resulting workload is therefore intentionally challenging: it isolates differences among lower-bounds rather than estimating their average performance over an unrestricted query workload.

\subsection{Evaluation of Ukkonen’s Algorithm for SED Lower-Bound}
\label{sec:exp_ukkonen}

In the first set of experiments we examine the impact of Ukkonen’s algorithm for SED lower-bound computation within a given threshold.
Since the SED lower-bound computes edit distance between preorder and postorder traversals, quadratic runtime may dominate.
Ukkonen’s algorithm calculates the edit distance in time $O(\tau \cdot \min(m,n))$ when only testing against a threshold $\tau$,
which is particularly efficient for small thresholds.
We show that the use of Ukkonen’s algorithm significantly stabilises the runtime of the SED lower-bound and enables scaling.

\begin{figure}
  \centering
  \includegraphics[width=0.35\textwidth]{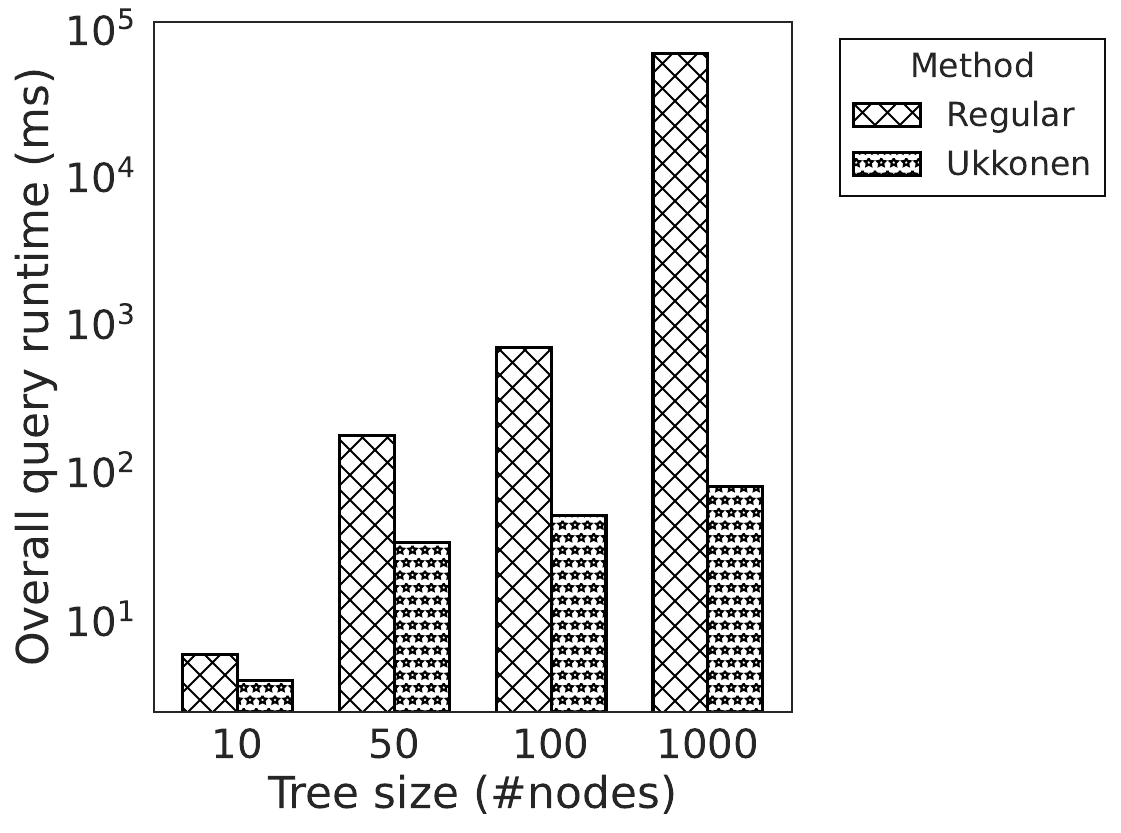}
  \caption{Comparison of SED lower-bound with and without Ukkonen's algorithm for varying tree size}
  \label{fig:SEDukkonen}
\end{figure}

Figure~\ref{fig:SEDukkonen} presents the measured SED lower-bound runtimes on \textit{Synthetic\textsubscript{size}} collection. Ukkonen's algorithm consistently outperforms the traditional SED computation when operating under a threshold constraint. This performance advantage becomes increasingly pronounced as tree size grows, with the improvement reaching several orders of magnitude for larger trees. Given this consistent superiority in threshold-based computation, all subsequent experiments exclusively utilize the Ukkonen-based implementation of SED lower-bounds.

\subsection{Evaluation of Different Tree Traversals}
\label{sec:treetraversal}

As discussed in Section~\ref{sec:tree_traversal}, different types of tree traversals can be applied in the SED lower-bound. The classic selection is a preorder and postorder~\cite{guha2002approximate}. We evaluate the following traversal combinations for both the SED lower-bound and the SED-struct filter:
\begin{itemize}
  \item Preorder and Postorder (pre/post).
  \item Reverse Postorder and Postorder (rpost/post).
  \item Reverse Postorder and Preorder (rpost/pre).
  \item Reverse Postorder and Reverse Preorder (rpost/rpre).
  \item Reverse Preorder and Postorder (rpre/post).
  \item Reverse Preorder and Preorder (rpre/pre).
\end{itemize}



To systematically evaluate which traversal combinations provide the best filtering performance, we computed precision across a total of 6 combinations of traversal types a mix of datasets.
The goal was to identify which traversal pairs yield the tightest lower-bounds and are the most optimal.

The datasets used for this evaluation include \textit{L4} dataset from the Synthetic\textsubscript{label} collection, \textit{RNA} and \textit{Sentiment} datasets from the real-world collection.

\subsubsection{SED Traversal Results}

Figure~\ref{fig:traversals_precision_sed} presents the precision of different traversal combinations for the SED lower-bound across selected datasets.
The results indicate that the SED lower-bound is not improved by alternative traversal combinations.
The highest precision is consistently achieved by the classic pre/post pair,
as well as the rpost/post, rpost/rpre, and rpre/pre combinations.
The remaining rpost/pre and rpre/post underperform. This is due to the fact that rpost/pre traversal sequences are mutually reversed. Because the SED between two strings is invariant under reversal -- for instance, if $\operatorname{SED}(S_1, S_2) = 2$, then $\operatorname{SED}(\operatorname{reverse}(S_1), \operatorname{reverse}(S_2)) = 2$ -- the distance between these sequences remains identical, in other words:
\[
  \operatorname{SED}(\operatorname{rpost}(T_1), \operatorname{rpost}(T_2))
  = \operatorname{SED}(\operatorname{pre}(T_1), \operatorname{pre}(T_2))
\]
Consequently, this combination fails to provide additional structural information,
leading to a less effective lower-bound.
The same applies to rpre/post.

\begin{figure}[htb]
  \centering
  \includegraphics[width=0.9\linewidth]{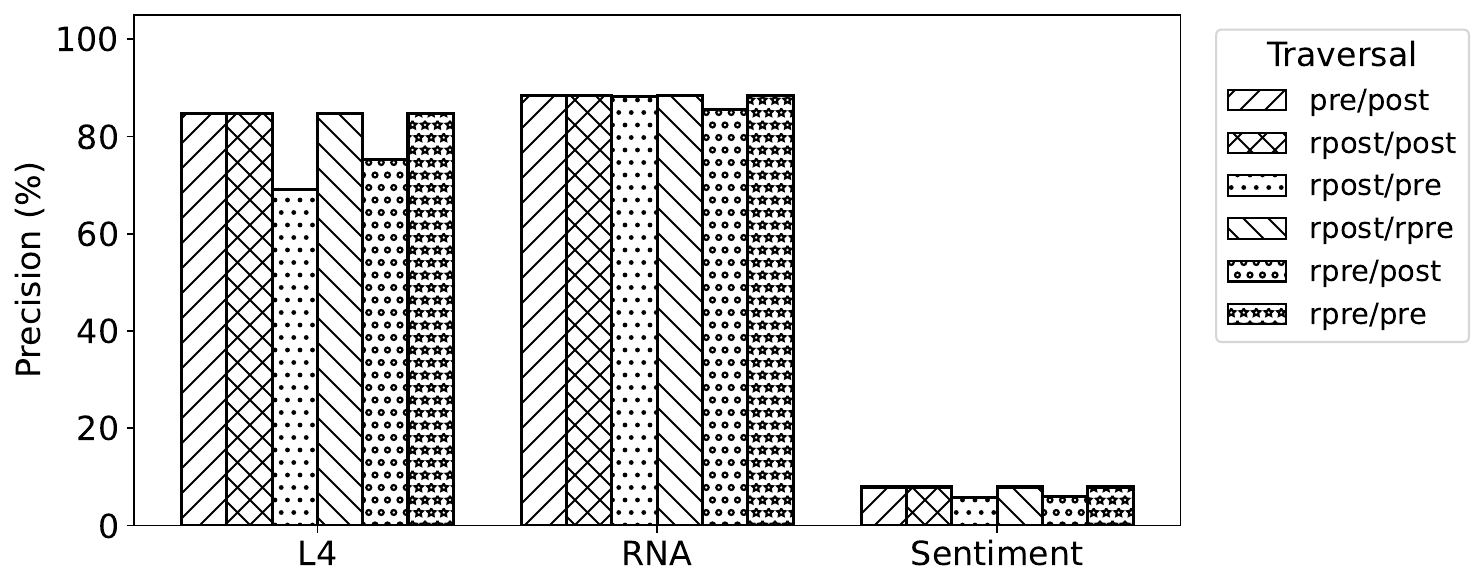}
  \caption{SED lower-bound precision for different tree traversal combinations}
  \label{fig:traversals_precision_sed}
\end{figure}

\subsubsection{SED-struct Traversal Results}

Figure~\ref{fig:traversals_precision_sedstruct} shows the corresponding results for the SED-struct filter.
The variance between different traversal combinations is higher and each combination performs differently because of the structural component of the SED-struct score.
Each traversal captures a different subset of structural axes (see~\ref{def:sed_struct_score}).
Overall, the best performing traversal combinations are rpre/pre. The assumption is that the structural axis of subtree size has higher impact
on the lower-bound precision rather than the ancestors axis.
Therefore, we always use rpre/pre tree traversal in the following experiments.

\begin{figure}[htb]
  \centering
  \includegraphics[width=0.9\linewidth]{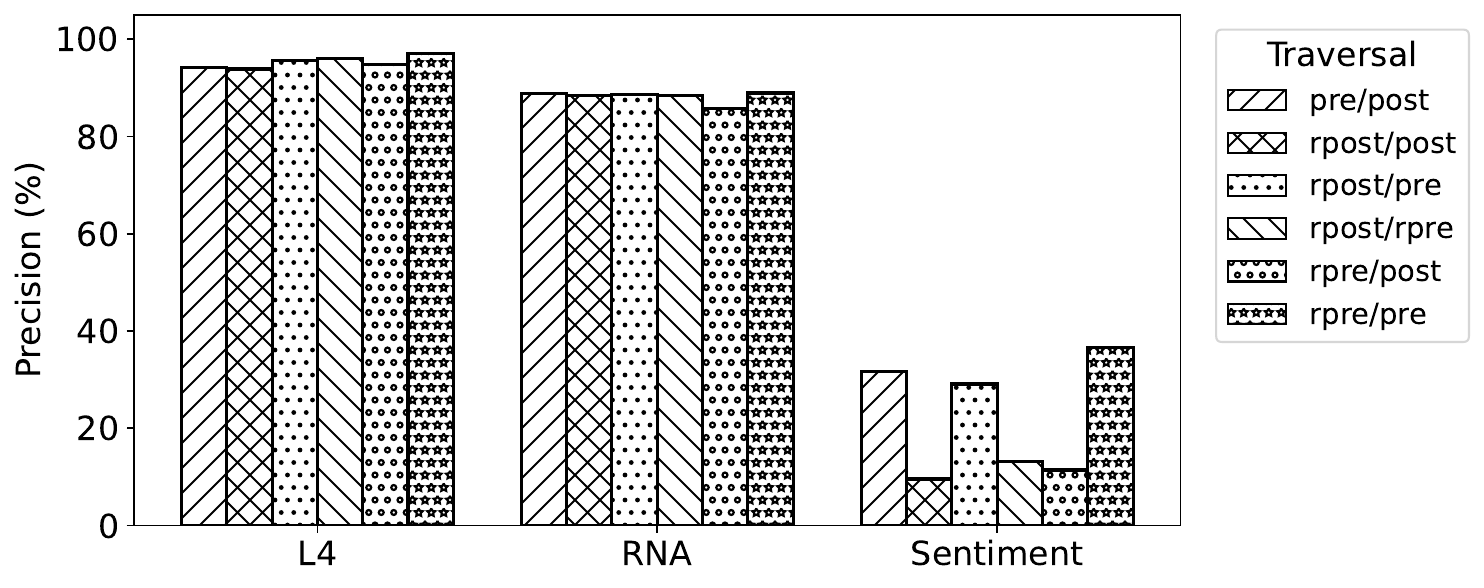}
  \caption{SED-struct filter precision for different tree traversal combinations}
  \label{fig:traversals_precision_sedstruct}
\end{figure}


\subsection{Comparison of Filtering Methods on Real Datasets}
\label{sec:exp_real_datasetes}

We evaluate the following filtering methods, as described in Section~\ref{sec:lb-filters}:
\begin{itemize}
  \item \emph{Binary branch bound}~\cite{yang2005similarity},
  \item \emph{Intersection lower-bound}~\cite{kailing2004efficient},
  \item \emph{Structural lower-bound}~\cite{hutter2019effective},
  \item \emph{SED lower-bound}~\cite{guha2002approximate},
  \item \emph{SED-struct threshold filter} (proposed in this paper).
\end{itemize}

The SED lower-bound and SED-struct filter implementations both utilize Ukkonen's algorithm. 
In accordance with the filter-verify framework described earlier, the search process for each query begins with an initial filtering stage to generate a set of potential \textit{candidates}. 
Every candidate in this set must subsequently undergo a verification phase to determine the exact TED. 
While the computational cost of filtering is generally negligible, the verification of these candidates represents the primary bottleneck in the total runtime.
Therefore, the major aspect influencing overall search runtime is filter precision.
To evaluate the performance of the filtering methods, we measure the overall search execution
time and filter precision, defined as the ratio of candidates that
are true positives.

\begin{figure}[htb]
  \centering
  \includegraphics[width=1\linewidth]{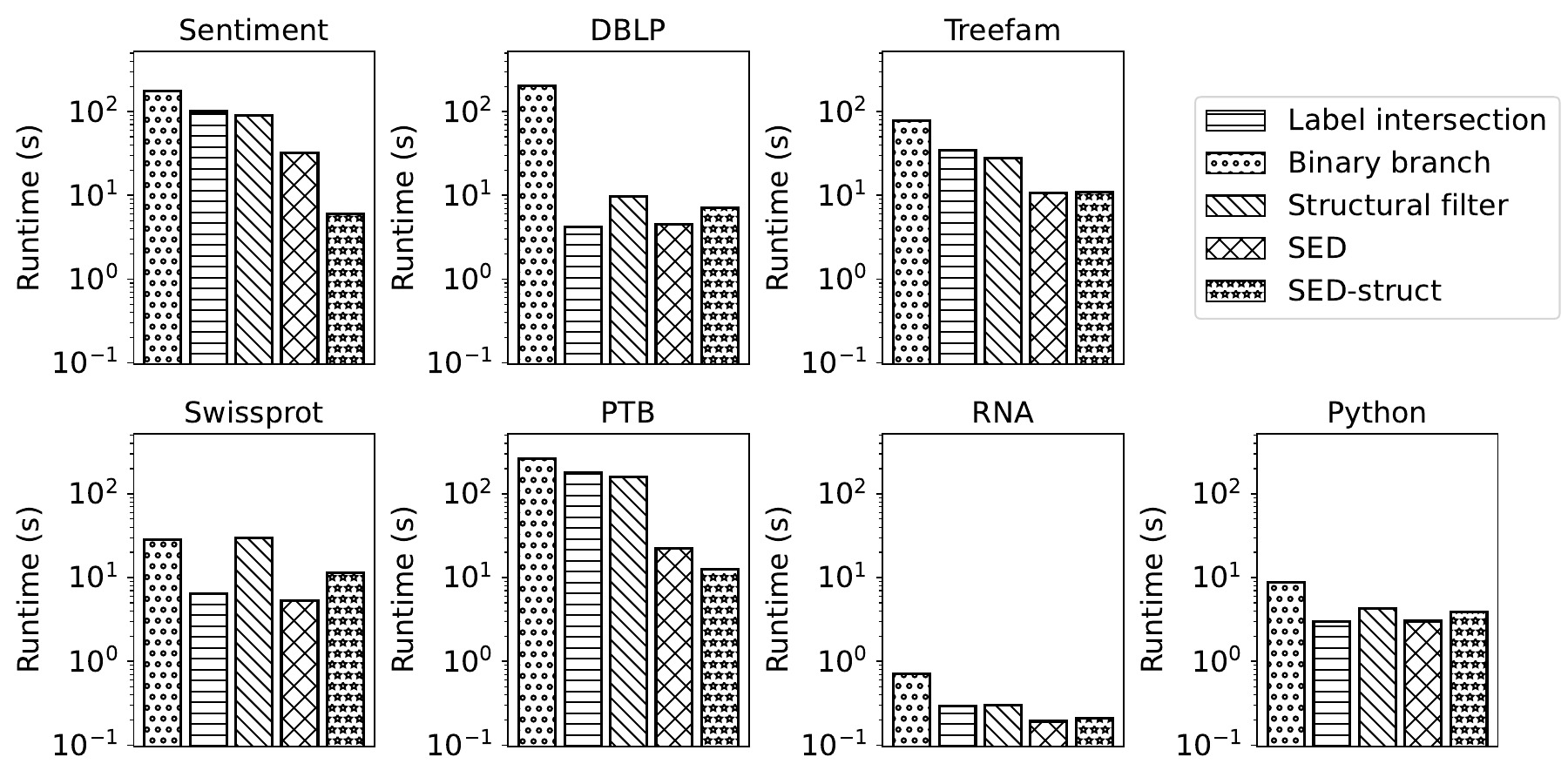}
  \caption{Overall query runtimes}
  \label{fig:filter_times}
\end{figure}

\subsubsection{Overall Query Runtime}

The overall query runtimes are illustrated in Figure~\ref{fig:filter_times}.
The SED-struct filter significantly outperforms all the other approaches on the \textit{Sentiment} and \textit{PTB} dataset. That is influenced mainly by
its superior precision as shown on Figure~\ref{fig:real_precision}.

\subsubsection{Filter Precision}

Figure~\ref{fig:real_precision} shows the precision of each filtering method.
These results mirror the runtime performance we saw earlier: the SED-struct filter achieves the highest precision across all
datasets, which directly reduces the work needed during verification.

The SED lower-bound is already highly precise on several datasets, which naturally limits the potential for further improvement by the SED-struct filter.
Since the SED-struct filter is an extension of SED, there is less room for additional pruning in cases where SED already identifies the true positives effectively. This behavior is particularly evident on the \textit{Swissprot} and \textit{DBLP} datasets. On these specific datasets, because the precision gains are minimal, the additional computational overhead required for the structural $L_1$ distance can result in a minor increase in latency. However, this added overhead is very little and does not
compromise the SED-struct filter's overall efficiency.

\begin{figure}[htb]
  \centering
  \includegraphics[width=1\linewidth]{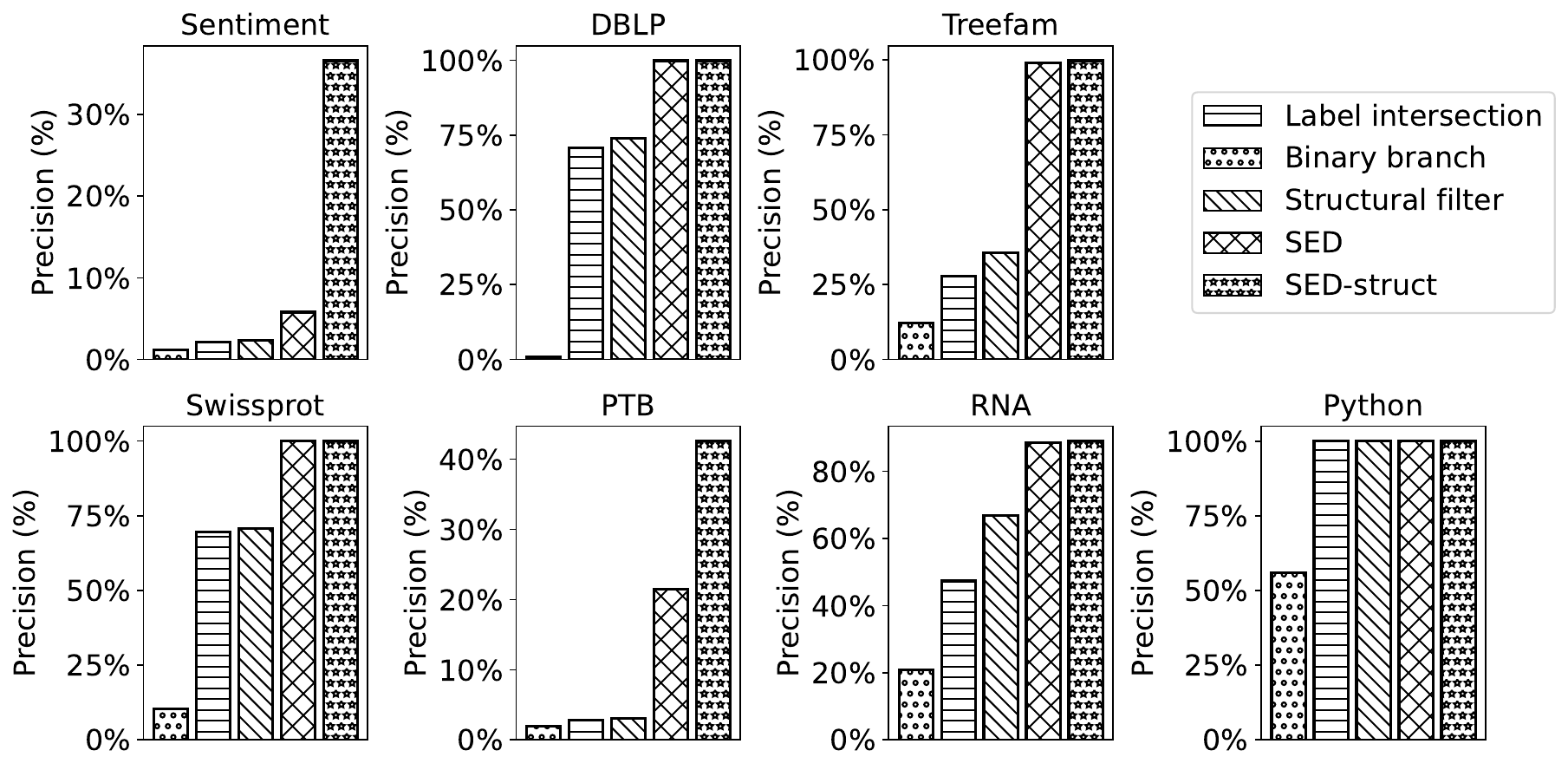}
  \caption{Precision on real datasets}
  \label{fig:real_precision}
\end{figure}

\subsubsection{PostgreSQL extension}
All above mentioned lower-bounds were implemented into a practically usable PostgreSQL extension~\footnote{\url{https://github.com/LukMRVC/tree_similarity_extension}},
capable of running similarity search queries. The same query set that was used for in-memory benchmarking
was also tested in this DBMS, tuned for real-world use cases.
The results in PostgreSQL 
\iffullversion
(Figure~\ref{fig:postgres_runtime})
\else
\fi
generally mirror the in-memory implementation, with the SED-struct filter consistently achieving the highest precision and lowest overall query runtime.

\iffullversion
\begin{figure}[htb]
  \centering
  \includegraphics[width=1\linewidth]{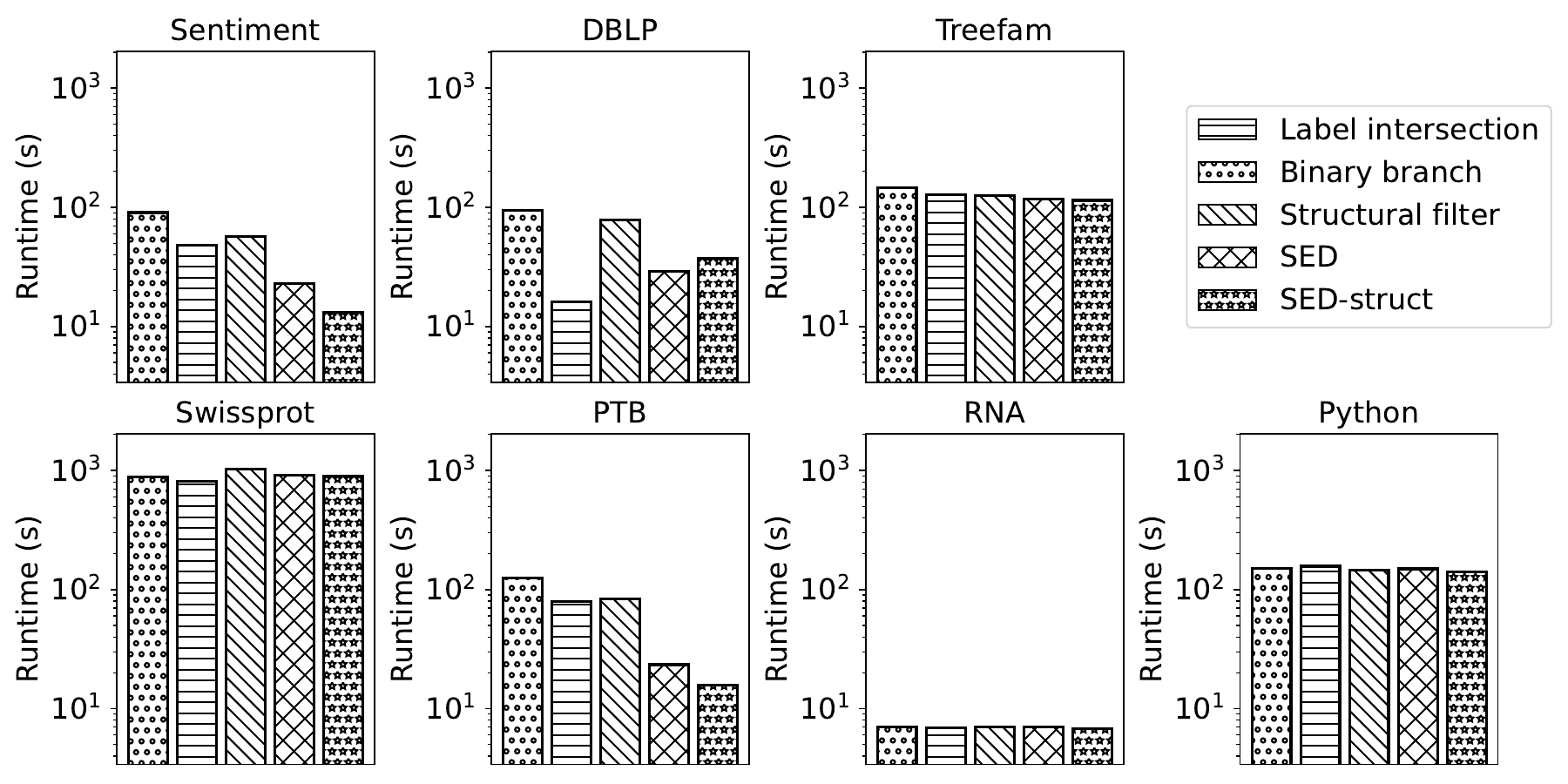}
  \caption{PostgreSQL experiments runtime}
  \label{fig:postgres_runtime}
\end{figure}
\else
\fi

\subsection{Analysis of SED-struct versus SED}

The aim of this experiment is to determine when the SED-struct filter outperforms the standard SED lower-bound. The goal is to identify conditions under which structural information provides a significant benefit.

We use two synthetic collections: Synthetic\textsubscript{struct} and Synthetic\textsubscript{label}. Both collections have a tiny number of different labels and a near-constant tree size. The details of creating these synthetic datasets
were discussed in section~\ref{subsec:datasets}.
Table~\ref{tab:datasets} illustrates that for both synthetic collections, the pool of available labels is restricted to a maximum of 10 distinct entries. This restriction was done on the basis of our prior experiments.
When the pool of available labels increases, the efficiency of SED algorithm drastically increases with it on synthetic collections, which makes it hard to identify how structural differences in trees affect the resulting precision
of the SED-struct filter.

\begin{figure}[htb]
  \centering
  \includegraphics[width=1\linewidth]{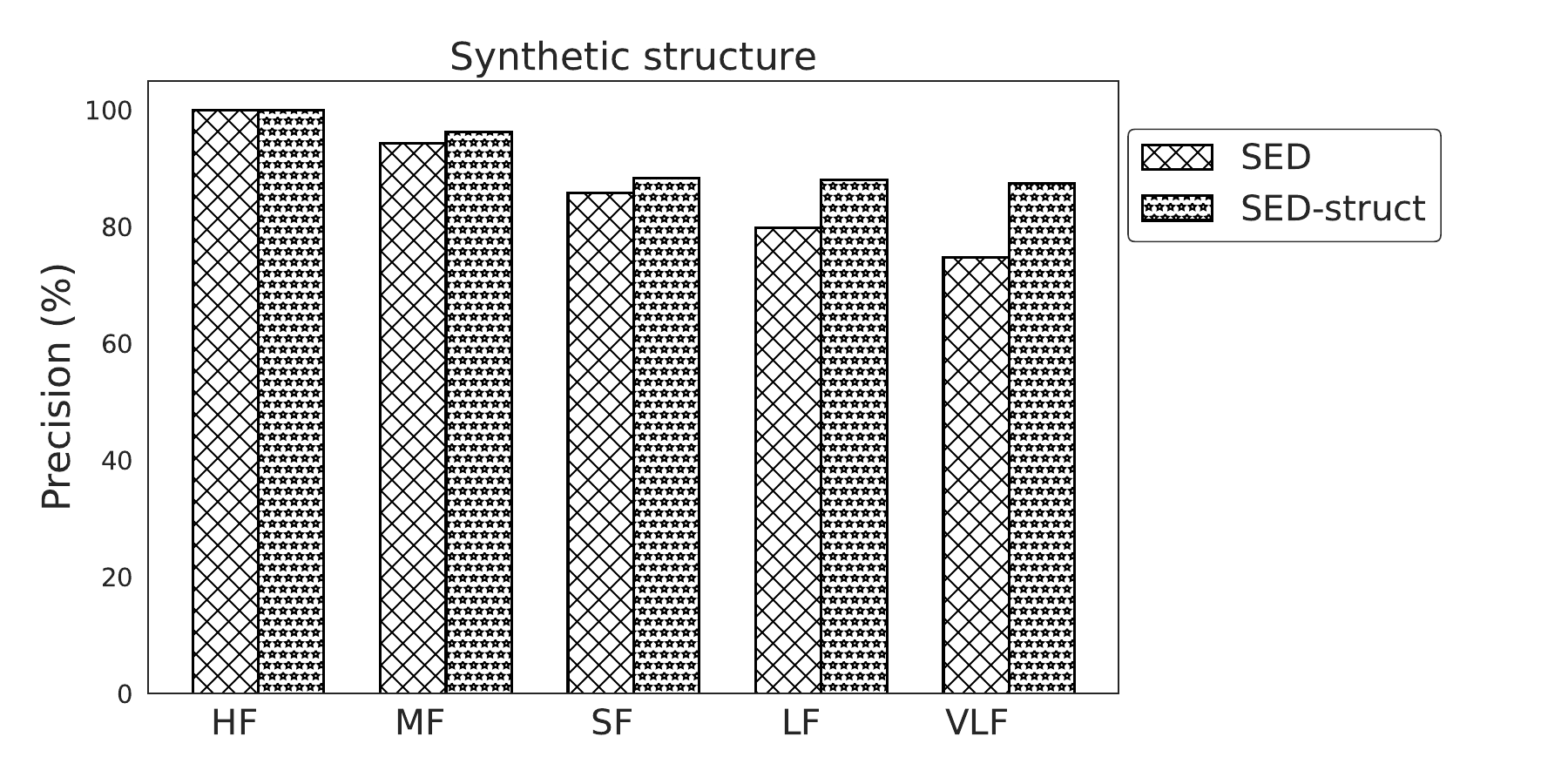}
  \caption{Precision on Synthetic structure datasets}
  \label{fig:synthetic_precision_struct}
\end{figure}

Figure~\ref{fig:synthetic_precision_struct} displays the results of the precision benchmarks across Synthetic\textsubscript{struct} collection,
where tree structures vary from high-fanout (\textit{HF}) to very low-fanout (\textit{VLF}) datasets.
As the average fanout decreases, the precision of both methods gradually declines. However, the SED-struct filter consistently
maintains higher precision than SED across all datasets. While the baseline SED lower-bound's precision drops, the SED-struct filter remains robust, preserving the slack. This gap confirms that incorporation of structural information is beneficial and results in
a more selective SED-struct filter.

The results visible in Figure~\ref{fig:synthetic_precision_label} highlight the resilience of the SED-struct filter
against a reduction in distinct labels, where the number of unique labels in datasets ranges from 10 (\textit{L10}) down to 2 (\textit{L2}).
The SED lower-bound demonstrates a high sensitivity to label diversity, as its precision significantly decreases, from 98\%
to 40\% when the label pool is reduced. In contrast, the SED-struct filter exhibits stability, maintaining a precision of approximately
95\% even for the \textit{L2} dataset. These results indicate that the SED-struct filter effectively utilizes structural information to compensate for low label variety,
ensuring high filtering efficacy even in label-poor environments.

\begin{figure}[htb]
  \centering
  \includegraphics[width=1\linewidth]{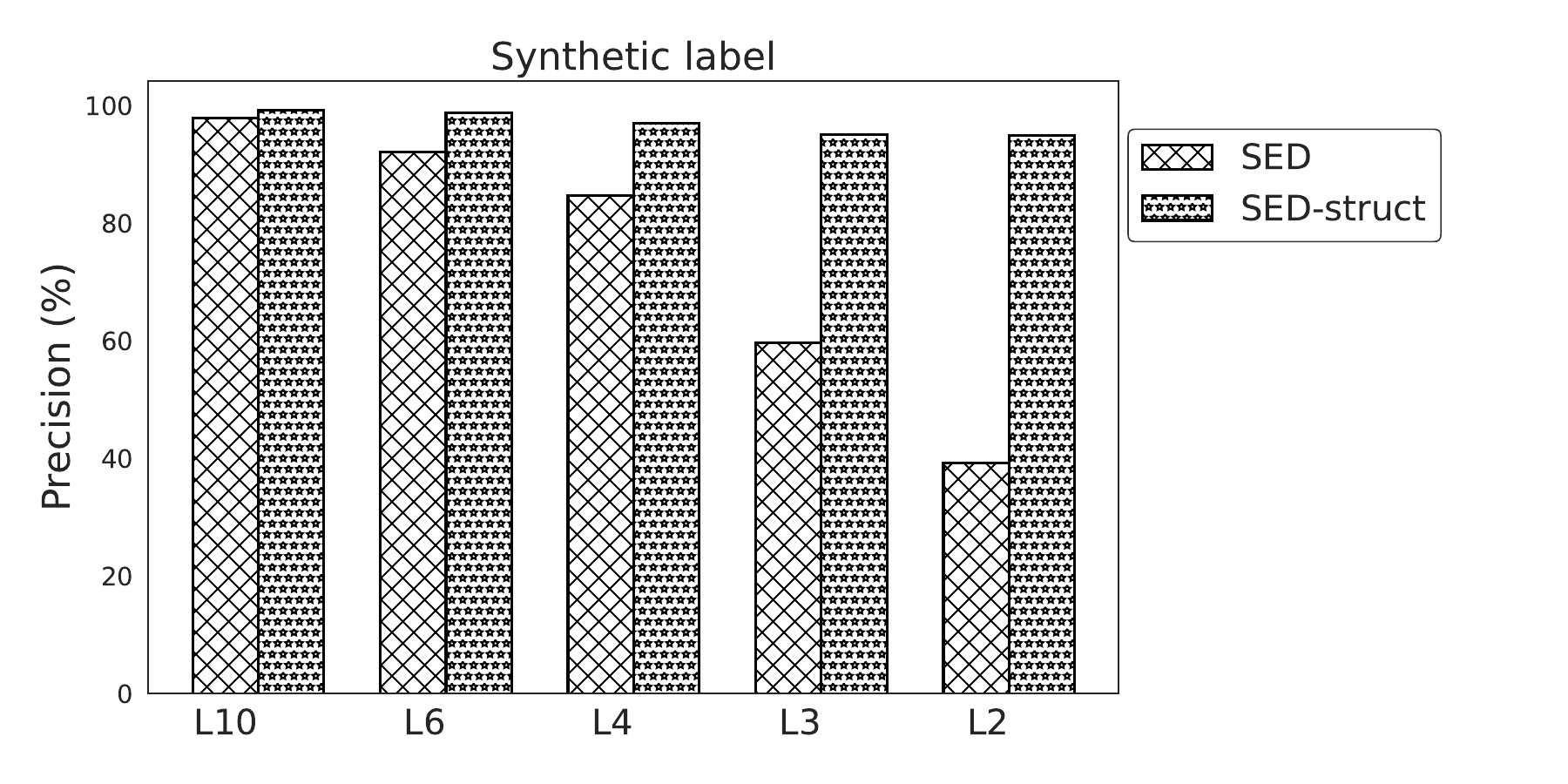}
  \caption{Precision on Synthetic label datasets}
  \label{fig:synthetic_precision_label}
\end{figure}

\subsection{Combining Filters}
\label{sec:combining_lb}

\iffullversion
Prior work on tree similarity search has advocated for combining multiple lower-bound filters to achieve an optimal balance between filtering speed and precision~\cite{li2014survey}. The rationale behind this approach is straightforward:
simpler lower-bounds (e.g., binary branch bound, intersection bound) can be computed extremely quickly but provide loose
approximations, while more sophisticated bounds require additional computation. By applying fast filters first and
progressively refining the candidate set with tighter bounds, one can theoretically achieve better overall performance
than using any single lower-bound in isolation.

However, our experimental results challenge this idea in the context of optimized filter implementations. Figure~\ref{fig:filter_query_times} presents the filtering times for all evaluated methods across real-world datasets. The results reveal that our SED and SED-struct implementations, utilizing Ukkonen's algorithm with the Berghel-Roach optimization, exhibit filtering times that are at most one order of magnitude slower than the simplest lower-bounds, while simultaneously achieving dramatically superior precision (as shown in Figure~\ref{fig:real_precision}).
The modest increase in filtering time is more than compensated by the substantial reduction in verification costs that
stems from higher precision.

\begin{figure}[htb]
  \centering
  \includegraphics[width=1\linewidth]{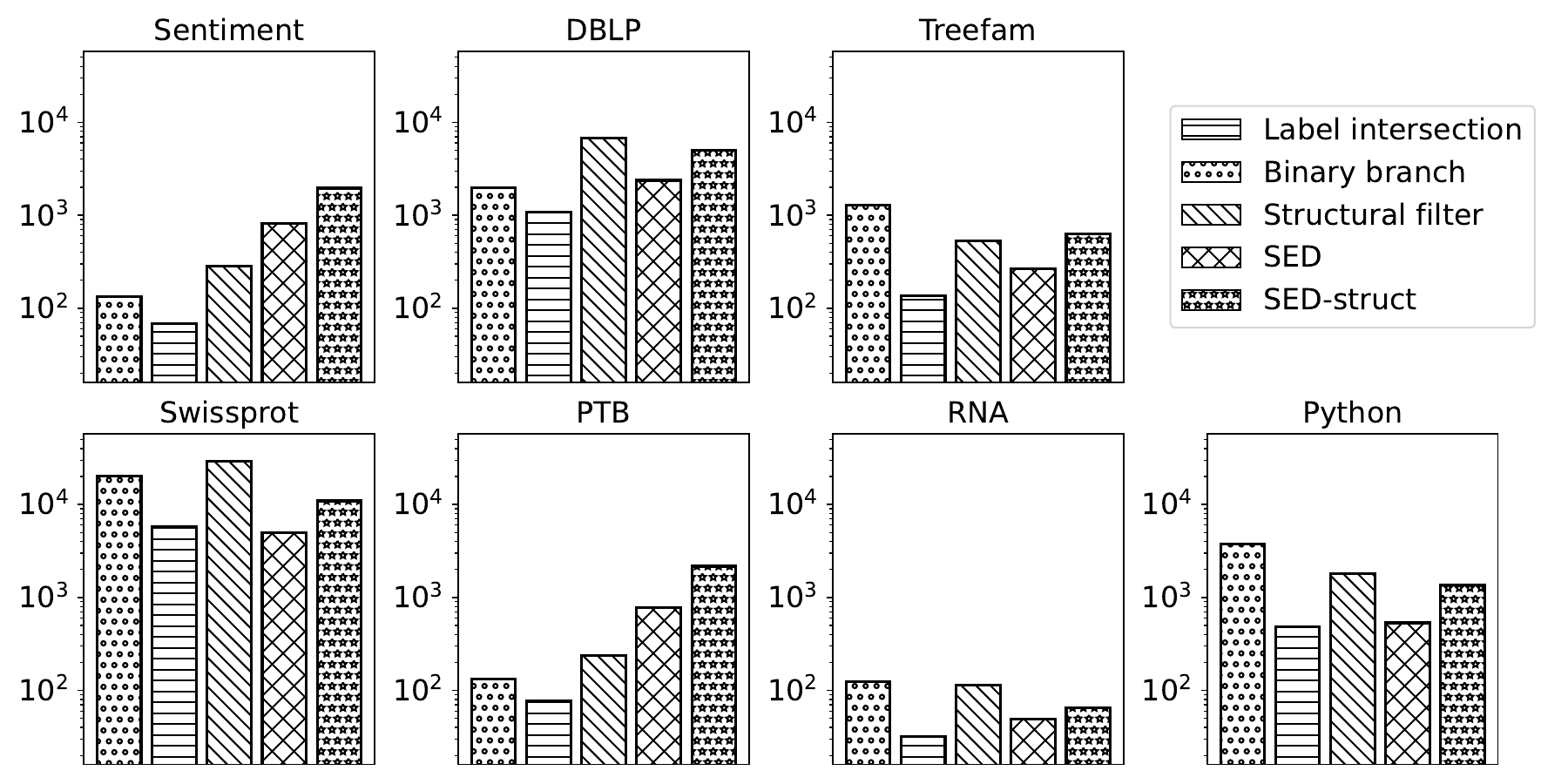}
  \caption{Filtering times for different methods on real datasets}
  \label{fig:filter_query_times}
\end{figure}

This finding has important practical implications. When the filtering phase represents only a small fraction of total
search time---as is the case with our optimized SED based implementations---the benefits of multi-stage filtering
diminish considerably. The overhead of applying multiple filters can outweigh the gains. Moreover,
since the SED lower-bound provides the highest precision, applying multiple simple lower-bounds won't significantly
improve the resulting overall filtering precision.
\fi


%% file: conclusion.tex
\section{Conclusion}

This paper addressed the problem of efficient threshold-based tree edit distance computation through lower-bounds. We presented three contributions advancing the state of the art.

First, we conducted a comprehensive experimental comparison of existing TED lower-bounds---including size, label intersection, binary branch, structural, and SED lower-bounds---evaluating their trade-offs between precision and computational cost across both real-world and synthetic datasets.

Second, we demonstrated that incorporating Ukkonen's bounded string edit distance algorithm into the SED lower-bound yields substantial runtime improvements, reaching several orders of magnitude for larger trees, while preserving the same pruning power. This optimization makes the SED lower-bound practical even for large tree collections.

Third, we proposed the SED-struct threshold filter, a novel extension of SED that integrates structural information through a structural guard mechanism. We formally proved its soundness: whenever the SED-struct filter rejects a tree pair at threshold $\tau$, the tree edit distance of the pair is greater than $\tau$. By exploiting positional relationships between tree nodes---specifically ancestors, descendants, preceding, and following axes---the SED-struct filter achieves higher precision than all other evaluated filtering methods. Our experiments on synthetic datasets show that the filter is particularly effective when label diversity is low: while the standard SED precision drops to approximately 40\% with only two distinct labels, the SED-struct filter maintains precision around 95\%. The choice of tree traversal also matters; we showed that the reverse preorder/preorder combination yields the best results for the SED-struct filter.

Our experimental evaluation further revealed that, with optimized Ukkonen-based implementations, the common strategy of cascading multiple simple lower-bounds offers little benefit. 

To support reproducibility and further research, we make our implementation and all experimental scripts publicly available. Our results can be independently verified and extended by the community. We also provide a PostgreSQL extension that makes the evaluated filtering methods available directly within a database system.


%% file: acknowledgments.tex
\section{Acknowledgments}

This work was supported by SGS, V\v{S}B -- Technical University of Ostrava, Czech Republic,
under grant no.\ SP2026/009 "Advanced big data processing".